\documentclass{article}
\usepackage[utf8]{inputenc}
\usepackage[margin=1in, includefoot, footskip=0.5in]{geometry}
\usepackage{natbib}

\usepackage{amsmath, amssymb, amsthm}
\usepackage{graphicx}
\usepackage[dvipsnames]{xcolor}
\usepackage{tikz}
\usepackage{bm}
\usepackage{proof}
\usepackage{amsfonts}
\usepackage{tcolorbox}
\usepackage{ascmac}
\usepackage{fancybox}
\usepackage[all]{xy}
\usepackage{varwidth}
\usepackage{etoolbox}
\usepackage{siunitx}
\usepackage{physics}
\usepackage{float}
\usepackage{bbm}

\usepackage{float}
\AtBeginDocument{\RenewCommandCopy\qty\SI}  %% Ensures the compiler directive is applied only once

\usepackage{apalike}
\usepackage{url}
\usepackage{hyperref}
\hypersetup{
  colorlinks=true,
  linkcolor=blue,
  filecolor=magenta,
  urlcolor=cyan,
  pdftitle={Overleaf Example},
  pdfpagemode=FullScreen,
}
\hypersetup{
    colorlinks,
    citecolor=black,
    filecolor=black,
    linkcolor=black,
    urlcolor=black
}

\usepackage{enumerate}
\usepackage[colorinlistoftodos]{todonotes}  %% Included only once

\usepackage[edges]{forest}
\usepackage{tree}
\usepackage{ulem}

\newtheorem{definition}{Definition}[section]
\newtheorem{theorem}{Theorem}[section]
\newtheorem{proposition}{Proposition}[section]

\newtheorem{assumption}{Assumption}[section]

\theoremstyle{definition}
\AtEndEnvironment{remark}{\qed}
\newtheorem{remark}{Remark}[section]

\newcommand{\yuyaf}[1]{\todo[inline,color=orange!50]{\textbf{For Yuya: }#1}}

\newcommand{\uslater}[1]{\todo[inline,color=pink!50]{\textbf{For us to think later: }#1}}

\newcommand{\indep}{\perp \!\!\! \perp}
\newcommand{\att}{{\theta_{\text{ATT}}}}
\newcommand{\naive}{{\theta^{\text{na\"ive}}}}

\newcommand{\attlu}{{\theta_{\text{ATT}}^{\text{LU}}}}
\newcommand{\attec}{{\theta_{\text{ATT}}^{\text{ECB}}}}

\newcommand{\attdid}{{\theta_{\text{ATT}}^{\text{DID}}}}
\newcommand{\attm}{{\theta_{\text{ATT}}^{\text{M}}}}
\newcommand{\attdidm}{{\theta_{\text{ATT}}^{\text{DIDM}}}}
\newcommand{\E}{{\text{E}}}

\begin{document}
\allowdisplaybreaks

\title{\LARGE %Matching Versus Difference in Differences
Matching $\leq$ Hybrid $\leq$ Difference in Differences\thanks{\setlength{\baselineskip}{4mm}First arXiv date: November 12, 2024. We would like to thank Raj Chetty for generously allowing us to use his data sets for our empirical application, and Jeff Smith and Petra Todd for their repeated guidance on replicating their papers. We benefitted from useful comments by  Raj Chetty, Jamie Fogel, Hidehiko Ichimura, and Matt Staiger. All remaining errors are ours.\smallskip}}
\author{
Yechan Park\thanks{\setlength{\baselineskip}{4mm}Opportunity Insights, Harvard University, 1280 Massachusetts Avenue, Cambridge, MA 02138 Email: \texttt{yechanpark@fas.harvard.edu}\smallskip}
\and
Yuya Sasaki\thanks{\setlength{\baselineskip}{4mm}Brian and Charlotte Grove Chair and Professor of Economics. Department of Economics, Vanderbilt University, VU Station B \#351819, 2301 Vanderbilt Place, Nashville, TN 37235-1819 Email: \texttt{yuya.sasaki@vanderbilt.edu}}
}
\date{}
%\listoftodos
\maketitle
\begin{abstract}\setlength{\baselineskip}{6mm}

% Since \citeauthor{lalonde1986evaluating}'s (\citeyear{lalonde1986evaluating}) seminal paper, there has been considerable interest in the economics literature in credibly estimating the short-term average treatment effect on the treated (ATT) using two time periods: pre- and post-intervention. Historically, scholars have used the experimental estimate as a baseline truth to explore whether Matching, Difference-in-Differences (DID), or their hybrid forms, could more accurately replicate this experimental benchmark \citep[e.g.,][]{heckman1998_2matching,dehejia2002propensity,smith2005does}.
% We revisit these methodologies through the lens of bracketing \textit{a la} \citet[][Section 5]{angrist2009mostly}, and establish the novel inequality relationship: Matching $\leq$ Hybrid $\leq$ DID under plausible and interpretable assumptions like negative selection.
% The true causal parameter is bracketed between the Matching and DID estimates, provided either set of assumptions holds.
% Hence, when non-negative treatments are expected, DID tends to give optimistic estimates while matching tends to give conservative estimates.
% Using four data sets that have been employed for the evaluation of job training and educational programs \citep{lalonde1986evaluating,heckman1998characterizing,smith2005does,athey2020combining}, we show that this inequality relationship is robust.
Since \citeauthor{lalonde1986evaluating}'s (\citeyear{lalonde1986evaluating}) seminal paper, there has been ongoing interest in estimating treatment effects using pre- and post-intervention data. Scholars have traditionally used experimental benchmarks to evaluate the accuracy of alternative econometric methods, including Matching, Difference-in-Differences (DID), and their hybrid forms \citep[e.g.,][]{heckman1998_2matching,dehejia2002propensity,smith2005does}.
We revisit these methodologies in the evaluation of job training and educational programs using four datasets \citep{lalonde1986evaluating,heckman1998characterizing,smith2005does,chetty2014measuring1,athey2020combining}, and show that the inequality relationship, Matching $\leq$ Hybrid $\leq$ DID, appears as a consistent norm, rather than a mere coincidence.
We provide a formal theoretical justification for this puzzling phenomenon under plausible conditions such as negative selection, by generalizing the classical bracketing \citep[][Section 5]{angrist2009mostly}. Consequently, when treatments are expected to be non-negative, DID tends to provide optimistic estimates, while Matching offers more conservative ones.
%We validate this inequality relationship across four datasets commonly used in the evaluation of job training and educational programs \citep{lalonde1986evaluating,heckman1998characterizing,smith2005does,athey2020combining,chetty2014measuring1}, demonstrating its robustness under varying conditions.
%\yechany{DIDM plays no role here?}
%In addition to the classical parametric bracketing (Angrist and Pischke 2009), we show that utilizing a hybrid estimator—DID following Matching on lagged outcomes—under reasonable assumptions, can further refine the bounds on the true causal parameter. We also provide an algebraic answer to the recent open debate on whether to condition on pretreatment outcome, that were primarily numerical based.
%- Extending these findings, we find that our bracketing results hold 

\medskip\noindent
{\bf Keywords:} bias, difference in differences, educational program, job training program, matching.

%\medskip\noindent
%{\bf JEL Codes:}
\end{abstract}

\newpage
%\listoftodos
% \section{bibouroku: for us to think later}
% \begin{itemize}
%     \item 
% \end{itemize}

\begin{comment}
\section{general outline}
\begin{itemize}
    \item - Main Part 1: Short term bracketing theory: 

    - Introduce the DID and the matching and the hybrid version

    - Show the bracketing relationship (ofc, have to say the ding and li(2019) paper here, but should be fine)
    
    - Show the  situations where the DIDM bias is greater than 0 or less than 0
    \item - Main Part2: Short term bracketing applicaton
    - revisit the Heckman et al type paper's application (JTPA, Lalonde's NSW)
    \item - Main Part3: Long term bracketing theory ( 
    \item - Main Part4: Long term bracketing application
    - in addition to the project star and the star, we should also have the california gain application as well? (
    
\end{itemize}
\end{comment}

%%%%%%%%%%%%%%%%%%%%%%%%%%%%%%%%%%%%%%%%%%%%%%%%%%%%
\section{Introduction}\label{sec:introduction}
%%%%%%%%%%%%%%%%%%%%%%%%%%%%%%%%%%%%%%%%%%%%%%%%%%%%

Since the seminal work by \citet{lalonde1986evaluating}, there has been substantial interest in accurately estimating the short-term average treatment effects on the treated (ATT) and other causal parameters using panel data with two (pre-treatment and post-treatment) periods \citep[e.g.,][]{heckman1998_2matching,dehejia2002propensity,smith2005does}. This body of literature has focused on a debate concerning which method -- matching (M), difference-in-differences (DID), or their hybrid (DIDM) -- is more effective at replicating experimental estimates when using observational data.

In the current era, where the DID has once again captured the attention of empirical practitioners, we contribute to the ongoing debate by examining the relationships among these three estimands from both empirical and theoretical perspectives. Our analysis clarifies the conditions under which one method provides the most conservative estimates while another offers the most optimistic estimates. By doing so, we provide a more nuanced understanding of the relative strengths and limitations of each approach, enabling researchers to make more informed methodological choices based on the specific characteristics of their data.

We begin by highlighting a puzzling pattern consistently observed in the empirical data used in both the aforementioned studies and other seminal works in the economics of education and labor economics. Specifically, the inequality relationship, M $\leq$ DIDM $\leq$ DID, frequently appears as a consistent norm, rather than a mere coincidence, when assessing the effectiveness of educational and job training programs.

For example, Figure \ref{fig:chabe_signed} illustrates the \textit{biases} of the M, DIDM, and DID estimates for job training programs relative to experimental estimates perceived as benchmark truths.
These estimates are excerpted from two seminal papers: \citet[HIST;][]{heckman1998characterizing} and \citet[ST;][]{smith2005does}, which revisits the analysis of \citet{lalonde1986evaluating,dehejia1999causal,dehejia2002propensity}. Despite differences in the job training programs studied, data sets, and estimation methods, the inequality relationship, M $\leq$ DIDM $\leq$ DID, holds robustly across both papers.
This relationship remains robust even when considering alternative estimation methods, different programs, or other data sets, as will be demonstrated shortly in this paper.
%%%%%%%%%%%%%%%%%%%%%%%%%%%%%%%%%%%%%%%%%%%%%%%%%%%%
\begin{figure}[t]
\centering
\includegraphics[width=0.66\textwidth]{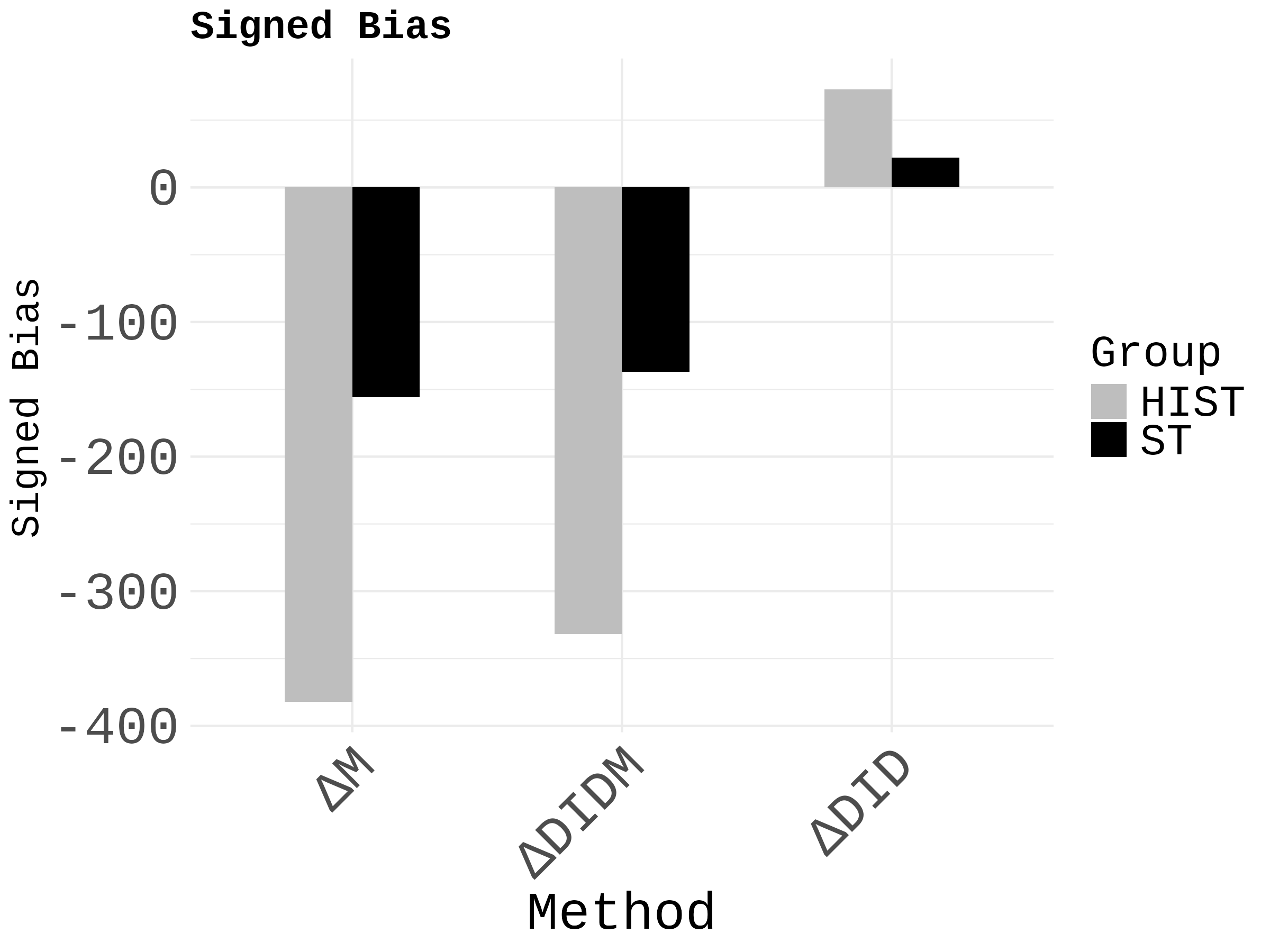}
\caption{The relative \textit{biases} of the M, DIDM, and DID estimates based on \citet{heckman1998characterizing} and \citet{smith2005does}.
The label HIST stands for Heckman-Ichimura-Smith-Todd while the label ST stands for Smith-Todd.
For each paper, the relative \textit{biases} of the observational estimates are displayed in percentage terms relative to an experimental estimate obtained using randomly allocated job training programs. 
Each estimation uses a set of auxiliary covariates including demographic characteristics. In addition, M and DIDM use pre-treatment earnings as a matching factor.}${}$
\label{fig:chabe_signed}
\end{figure}
%%%%%%%%%%%%%%%%%%%%%%%%%%%%%%%%%%%%%%%%%%%%%%%%%%%%

This pattern is reminiscent of the so-called `bracketing' relationship between the lagged dependent variable (LDV) estimator and the fixed-effect (FE) estimator in panel regressions, as presented in \citet[][Section 5]{angrist2009mostly}. Specifically, \cite{angrist2009mostly} outline plausible conditions under which the inequality relationship, LDV $\leq$ FE, holds on theoretical grounds.
Their assumptions necessitate negative selection into treatment and the absence of explosive outcomes, which make plausible sense in many labor economic settings.
This result has been elegantly extended to a nonparametric setup by \cite{ding2019bracketing}.

We find that similar plausible conditions, in the spirit of \cite{angrist2009mostly} and \cite{ding2019bracketing}, also give rise to the aforementioned inequality relationship, M $\leq$ DIDM $\leq$ DID.
In fact, a special case of our general double bracketing result, M $\leq$ DIDM $\leq$ DID, reduces to the conventional bracketing result, LDV $\leq$ FE, established by \citet{angrist2009mostly} and \citet{ding2019bracketing}.

Recall that M, DID, and DIDM identify the true causal parameter under the assumptions of observational unconfoundedness, parallel trends, and conditional parallel trends, respectively. In practice, an empirical researcher may not know which of these three alternative conditions is satisfied for an application of interest. Our double bracketing result, M $\leq$ DIDM $\leq$ DID, implies that the true causal parameter is bracketed below by M and above by DID, with DIDM between them, when one of the three alternative assumptions holds true.
In other words, our theoretical prediction implies that the DID approach tends to yield the most optimistic estimates while the M approach tends to yield the most conservative estimates. 

While our discussions primarily focus on the classic two-period framework \citep[e.g.,][]{heckman1998characterizing,smith2005does}, we also extend our double bracketing result to a broader class that includes cases with multi-dimensional covariates, dynamic and multi-period settings, and impulse response functions in event studies, as explored in the recent literature \citep[e.g.,][ among others]{callaway2018difference,acemoglu2019democracy,deChaisemartin2020two,dube2023local,imai2023matching}.

Using four data sets that have been used for the evaluation of job training and educational programs in the literature \citep{lalonde1986evaluating,heckman1998characterizing,smith2005does,athey2020combining}, we empirically examine our double bracketing hypothesis, M $\leq$ DIDM $\leq$ DID.
In light of the robustness of this inequality relationship in all the empirical scenarios, we provide formal theoretical explanations for this intriguing phenomenon.
Our assumptions required for the double bracketing relationship, motivated by \cite{angrist2009mostly} and \cite{ding2019bracketing} as mentioned earlier, is not only in line with the literature but are also empirically testable.
Hence, we examine our assumptions, as well as the double bracketing relations \textit{per se}, using these empirical data sets.
It turns out that our assumptions, as well as the double bracketing relationship, M $\leq$ DIDM $\leq$ DID, indeed hold robustly in all these empirical cases.

\section{Relation to the Literature}\label{sec:literature}
%%%%%%%%%%%%%%%%%%%%%%%%%%%%%%%%%%%%%%%%%%%%%%%%%%%%
The question we investigate relates to a long literature of econometrics including a number of seminal papers.

First, this paper closely relates to the classical debate in economics on what type of non-experimental estimates replicate the experimental estimates of \cite{lalonde1986evaluating}. 
The pioneering papers in that literature are \citet{heckman1998characterizing}, \citet{heckman1998_2matching}, \citet{dehejia2002propensity}, and \citet{smith2005does}.
There, the main purpose was to find the single best estimator, often measured by the \textit{absolute bias}. 
We are interested in \textit{signs}, as well as the magnitudes, of bias to identify which estimator is the most conservative/optimistic.

More recently, one notable work revisited a related problem. Specifically, Figure \ref{fig:chabe} replicates \citet[][Figure 3]{chabe2017should} and presents the \textit{absolute biases} of the three estimands, M, DIDM, and DID, based on the estimates from \citet{heckman1998characterizing} and \citet{smith2005does}. 
%%%%%%%%%%%%%%%%%%%%%%%%%%%%%%%%%%%%%%%%%%%%%%%%%%%%
\begin{figure}[tb]
\centering
\includegraphics[width=0.66\textwidth]{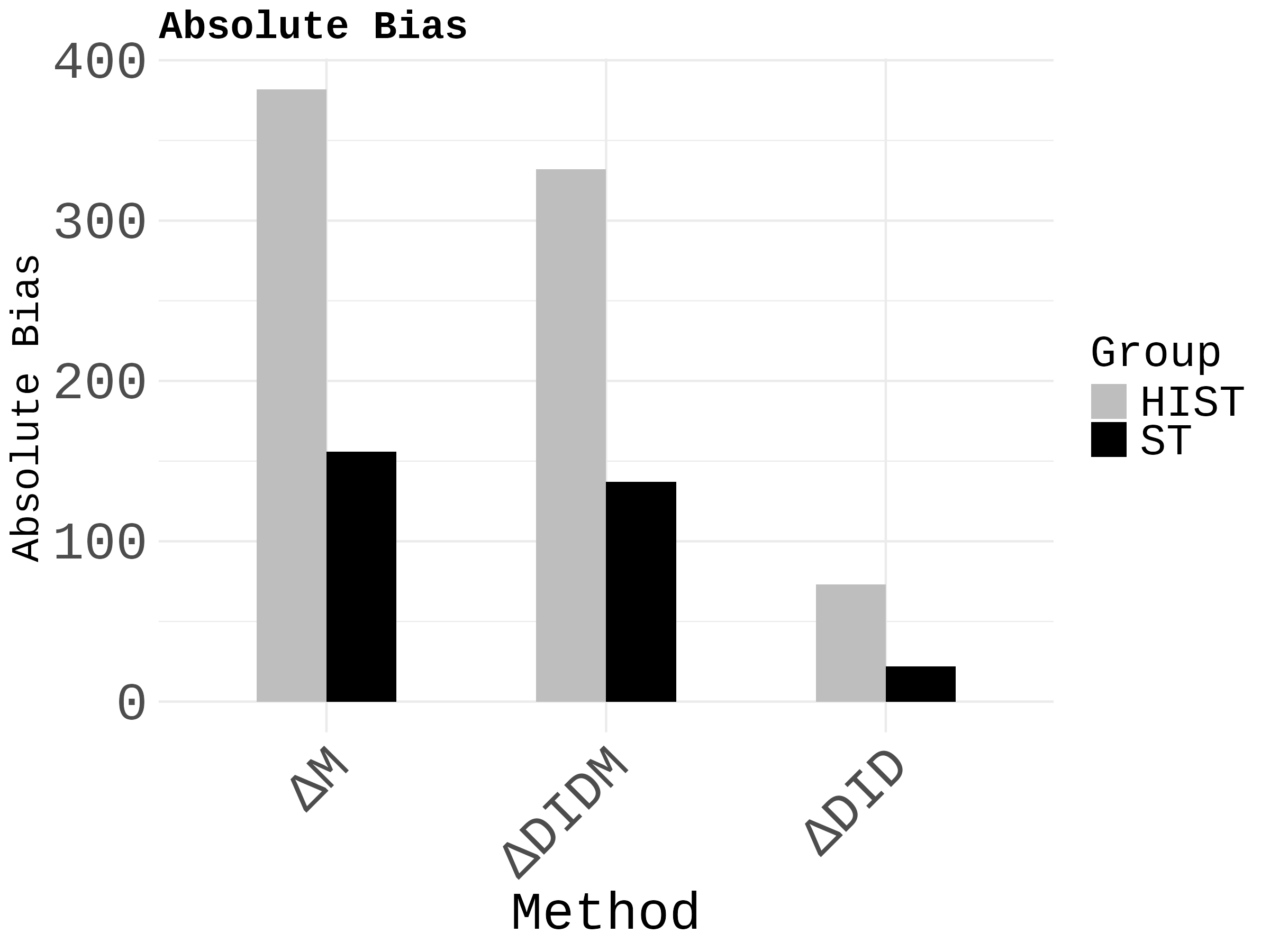}
\caption{Replication of Figure 3 in \cite{chabe2017should} based on the estimates by \citet{heckman1998characterizing} and \citet{smith2005does}.
The label HIST stands for Heckman-Ichimura-Smith-Todd while the label ST stands for Smith-Todd.
For each paper, the absolute \textit{biases} of M, DIDM, and DID estimates are displayed in percentage terms relative to an experimental estimate obtained using randomly allocated job training programs. 
Each estimate uses a set of auxiliary covariates including demographic characteristics.
In addition, M and DIDM use pre-treatment earnings as a matching factor.}${}$
\label{fig:chabe}
\end{figure}
%%%%%%%%%%%%%%%%%%%%%%%%%%%%%%%%%%%%%%%%%%%%%%%%%%%%
With the symbol `$\Delta$' denoting the \textit{bias}, this figure essentially shows $|\Delta $M$|\geq|\Delta $DIDM$|\geq|\Delta $DID$|$, implying that DID achieves the smallest bias for these empirical cases.
Notably, this Figure \ref{fig:chabe} depicts the \textit{absolute biases} corresponding to the \textit{signed biases} shown in Figure \ref{fig:chabe_signed} in our introductory section.

Thus, while the existing literature primarily focuses on identifying the best estimator based on absolute bias, our paper distinguishes itself by examining the relative signs of biases. This signed characterization provides more nuanced insights; for instance, Figure \ref{fig:chabe_signed} illustrates that M underestimates the true causal effects, whereas DID overestimates them.
Beyond this primary distinction between \textit{absolute} and \textit{signed} biases, our paper differs from \citet{chabe2017should} in two other key aspects. First, we provide a formal theoretical result, whereas the existing work derives its conclusions from numerical studies. Second, the existing work relies on functional form and structural assumptions -- such as separability and distinctions between transitory and permanent income -- tailored to job training programs. In contrast, we adopt a more model-free, nonparametric approach.

Our theoretical results build on the existing literature on bracketing. For linear panel models, \citet[][Section 5]{angrist2009mostly} develop a bracketing relationship, showing that LDV $\leq$ FE. This result has been extended to a nonparametric context by \citet{ding2019bracketing}.
We contribute to this literature in two key ways. First, we consider a more general framework that encompasses the symmetric DID (DIDM) studied by \citet{heckman1998characterizing} and \citet{smith2005does} in particular. Second, within this general framework, we show that DIDM, which \citet{heckman1998characterizing} recommend for minimizing bias, is further bracketed between M and DID, resulting in our double bracketing outcome.
A special case of our general framework reduces M and DIDM to LDV, and thus, our double bracketing relationship, M $\leq$ DIDM $\leq$ DID, simplifies to the conventional bracketing relationship, LDV $\leq$ FE.

Finally, this paper is of course related to the burgeoning literature on event studies and DID today -- see recent surveys by \citet{deChaisemartin2023two} and \citet{roth2023review} among others.
This paper is also related to the huge literature on matching, including covariate matching, propensity score matching, (augmented) inverse probability weighting, nearest neighbor methods, and regression adjustments, among others.
In particular, the matching literature has emphasized the importance of conditioning on lagged outcomes as matching factors to improve the precision of causal inference \citep{lalonde1986evaluating,dehejia1999causal,dehejia2002propensity}.
Economically, matching on lagged outcomes helps to address the so-called
``\citet{ashenfelter1978estimating} dip.''
The empirical work by \citet{acemoglu2019democracy} confirms that conditioning on lagged outcomes indeed yields more credible estimates in event studies.
Effectively, they advocate the M over the DID in our language.
Later, \citet[][Section 4.1]{dube2023local} generalize \citet{acemoglu2019democracy} to a hybrid framework which we refer to as the DIDM in this paper.
See the DID${}_\text{M}$ estimator of \citet{deChaisemartin2020two} and the (panel) matching estimator of \citet{imai2023matching}, as well as \citet[][Section 4.1]{dube2023local} -- they all propose and study the properties of what we refer to as the DIDM.
We contribute to the above literature by shedding light on the systematic relationship among M, DIDM, and DID, and empirically and theoretically examining this relationship.

%%%%%%%%%%%%%%%%%%%%%%%%%%%%%%%%%%%%%%%%%%%%%%%%%%%%
\section{The Setup and Definitions}
%%%%%%%%%%%%%%%%%%%%%%%%%%%%%%%%%%%%%%%%%%%%%%%%%%%%

Let $W$ denote the group of treatment assignment such that units with $W=1$ receive treatment between periods $t=0$ and $t=1$, while those with $W=0$ remain untreated.
Thus, everyone is untreated for $t \leq 0$, and only those with $W=1$ are treated for $t \geq 1$.
Let $Y_t(d)$ denote the potential outcome under treatment status $d$ at time $t$.
With these notations, suppose that a researcher is interested in identifying the average treatment effect on the treated (ATT) at time $t=1$ defined by
\begin{align*}
\att = E[ Y_1(1) - Y_1(0) | W=1].
\end{align*}

%%%%%%%%%%%%%%%%%%%%%%%%%%%%%%%%%%%%%%%%%%%%%%%%%%%%
\subsection{The M, DID, and DIDM Estimands}
%%%%%%%%%%%%%%%%%%%%%%%%%%%%%%%%%%%%%%%%%%%%%%%%%%%%

Letting $Y_t$ denote the observed outcome at time $t$, the seminal paper by \citet{heckman1998characterizing} proposes three alternative estimands to this goal:
\begin{align*}
\attm &= E[ Y_1 | W=1] - E[ E[Y_1 |W=0, Y_{-s} ] | W=1] \quad (s \geq 0),
\\
\attdid &= E[ Y_1 - Y_0 | W= 1] - E[ Y_1 - Y_0 | W=0],
\quad\text{and}\\
\attdidm &= E[E[ Y_1 -Y_0 | Y_{-s}, W=1] - E[ Y_1 -Y_0 | Y_{-s}, W=0]|W=1] \qquad (s \geq 0), 
\end{align*}
which we call the matching (M), the difference-in-differences (DID), and the difference-in-differences matching (DIDM), respectively.
Here, we focus on the past outcome $Y_{-s}=Y_{-s}(0)$ as the key matching criterion following \citet{heckman1998characterizing} and \citet{smith2005does}, but we present an extension to general vector-valued matching criteria $X$ in Section \ref{sec:general} to include other auxiliary covariates as well as more lagged outcomes as in the local projection approaches to event studies \citep[e.g.,][]{acemoglu2019democracy}.

The M estimand identifies the true ATT (i.e, $\attm = \att$ holds) if the matching condition
\begin{align*}
\text{Condition M:} \ \ \ (Y_1(1), Y_1(0) ) \indep W | Y_{-s}
\end{align*}
is satisfied.
The DID estimand identifies the true ATT (i.e, $\attdid = \att$ holds) if the parallel trend condition
\begin{align*}
\text{Condition DID:} \ \ \
\E[Y_1(0)-Y_0(0)|W=0]
=
\E[Y_1(0)-Y_0(0)|W=1]
\end{align*}
is satisfied.
Finally, the DIDM estimand identifies the true ATT (i.e, $\attdidm = \att$ holds) if the conditional parallel trend condition
\begin{align*}
\text{Condition DIDM:} \ \ \
\E[Y_1(0)-Y_0(0)|Y_{-s},W=0]
=
\E[Y_1(0)-Y_0(0)|Y_{-s},W=1]
\end{align*}
is satisfied.

As \citet{heckman1998characterizing} and \citet{smith2005does} stressed, 
in the absence of knowledge about the true data-generating process, 
there is a risk of bias associated with the three estimands, $\attm$, $\attdid$, and $\attdidm$.
For instance, when the true DGP satisfies Condition M but does not satisfy Condition DID or Condition DIDM, then only $\attm$ is guaranteed the identification while estimators for $\attdid$ and $\attdidm$ are doomed to be biased in general.
Hence, understanding the systematic relationship among 
$\attm$
$\attm$, and
$\attdid$
can help researchers form insights into the possible range within which the true causal effect, 
$\att$, may lie when one of the three alternative conditions holds.
We will empirically confirm the systematic double bracketing relationship, $\attm \leq \attdidm \leq \attdid$, in Section \ref{sec:empirical_double_bracketing} and provide theoretical explanations for it in Sections \ref{sec:simple}--\ref{sec:general}.

\subsection{Mutual Non-Nestedness of The M, DID, and DIDM Conditions}\label{sec:non_nested}
{\color{black}
The three identifying assumptions underlying the estimands \(\attm\), \(\attdid\), and \(\attdidm\) each restrict the data‐generating process (DGP) in a distinct way. No one assumption nests another; that is, none of the conditions implies or is implied by any of the others. In the following, we provide intuitive counterexamples—some leveraging differences in sample composition—to illustrate these distinctions.

\subsubsection{Condition M Holds but Condition DIDM Fails}

Consider the following data-generating process (DGP) for the potential outcomes:
\begin{align*}
Y_0(0) &= \mu(W)+\eta_0(0), & Y_1(0)&=\eta_1(0),\\
Y_0(1) &= \mu(W)+\eta_0(1), & Y_1(1)&=\eta_1(1),
\end{align*}
where $(\eta_0(0),\eta_0(1),\eta_1(0),\eta_1(1)) \indep (W, Y_{-s})$ and $\mu(\cdot)$ is an arbitrary function such that $\mu(0) \neq \mu(1)$.
Under this DGP, $(Y_1(1),Y_1(0)) \indep W | Y_{-s}$ holds, but $\E[Y_1(0)-Y_0(0)|Y_{-s},W=0]=\E[\eta_1(0)-\eta_0(0)]-\mu(0) \neq \E[\eta_1(0)-\eta_0(0)]-\mu(1) = \E[Y_1(0)-Y_0(0)|Y_{-s},W=1]$.
In other words, Condition M holds, but Condition DIDM fails.

\subsubsection{Condition DIDM Holds but Condition M Fails}

Conversely, consider the following DGP for the potential outcomes:
\begin{align*}
Y_0(0) &= \mu(W)+\eta_0(0), & Y_1(0)&= \mu(W) + \eta_1(0),\\
Y_0(1) &= \mu(W)+\eta_0(1), & Y_1(1)&= \mu(W) + \eta_1(1),
\end{align*}
where $(\eta_0(0),\eta_0(1),\eta_1(0),\eta_1(1)) \indep (W, Y_{-s})$ and $\mu(\cdot)$ is an arbitrary function such that $\mu(0) \neq \mu(1)$.
Under this DGP, $\E[Y_1(w)|Y_{-s},W=0]=\mu(0)+\E[\eta_1(0)] \neq \mu(1)+\E[\eta_1(0)] = \E[Y_1(w)|Y_{-s},W=1] $, and hence, $(Y_1(1),Y_1(0)) \not\indep W | Y_{-s}$, but $\E[Y_1(0)-Y_0(0)|Y_{-s},W=0]=\E[\eta_1(0)-\eta_0(0)] = \E[Y_1(0)-Y_0(0)|Y_{-s},W=1]$ holds.
In other words, Condition M fails, but Condition DIDM holds.

\begin{comment}
Conversely, suppose that the untreated outcome changes are homogeneous across treatment groups when conditioning on \(Y_{-s}\), so that
\[
E[Y_1(0)-Y_0(0)\mid Y_{-s},W=1] = E[Y_1(0)-Y_0(0)\mid Y_{-s},W=0] = \delta(Y_{-s}),
\]
for some function \(\delta(\cdot)\). However, let the level of the untreated outcome at time 1 differ across groups even after conditioning on \(Y_{-s}\). For instance, consider:
\[
\begin{aligned}
Y_1(0) &= \mu(Y_{-s}) + \zeta(W) + \epsilon_1, \\
Y_0(0) &= \mu(Y_{-s}) + \epsilon_0,
\end{aligned}
\]
with \(E[\epsilon_1-\epsilon_0\mid Y_{-s},W]=\delta(Y_{-s})\) and \(\zeta(1)\neq \zeta(0)\). In this case,
\[
E[Y_1(0)\mid Y_{-s},W=1] = \mu(Y_{-s}) + \zeta(1) \quad \text{and} \quad E[Y_1(0)\mid Y_{-s},W=0] = \mu(Y_{-s}) + \zeta(0),
\]
so that the matching condition
\[
(Y_1(1),Y_1(0)) \indep W\mid Y_{-s}
\]
fails due to the presence of the treatment-dependent term \(\zeta(W)\). Nonetheless, because \(\zeta(W)\) is constant over time, the difference
\[
E[Y_1(0)-Y_0(0)\mid Y_{-s},W] = \delta(Y_{-s})
\]
remains the same across groups. Hence, Condition DIDM holds even though Condition M is violated.
\end{comment}
\subsubsection{Condition DIDM Holds but Condition DID Fails}

Consider the following DGP for untreated potential outcomes:
\begin{align*}
Y_0(0) &= \mu(Y_{-1}) + \eta_0(0), & Y_1(0) &= \mu(Y_{-1}) + \eta_1(0),
\end{align*}
where $E[\eta_1(0)-\eta_0(0)|Y_{-1},W]=Y_{-1}$ and $E[Y_{-1}|W=0] \neq E[Y_{-1}|W=1]$.
Under this DGP, we have 
$E[Y_1(0)-Y_0(0)|Y_{-1},W=0]=Y_{-1}=E[Y_1(0)-Y_0(0)|Y_{-1},W=1]$, but
$
E[Y_1(0)-Y_0(0)|W=0] = E[E[\eta_1(0)-\eta_0(0)|Y_{-1},W=0]|W=0] = E[Y_{-1}|W=0]
\neq
E[Y_{-1}|W=1] = E[E[\eta_1(0)-\eta_0(0)|Y_{-1},W=1]|W=1] = E[Y_1(0)-Y_0(0)|W=1].
$
In other words, Condition DIDM holds, but Condition DID fails.   Intuitively, even if the parallel trends hold at every level of \(Y_{-1}\), differences in the distribution of \(Y_{-1}\) between treatment groups can lead to unequal unconditional trends.

\subsubsection{Condition DID Holds but Condition DIDM Fails}

Conversely, consider the following DGP for untreated potential outcomes:
\begin{align*}
Y_0(0) &= \mu(Y_{-1}) + \eta_0(0), & Y_1(0) &= \mu(Y_{-1}) + \eta_1(0),
\end{align*}
where $E[\eta_1(0)-\eta_0(0)|Y_{-1},W] = (2W-1)Y_{-1}$ and $E[Y_{-1}|W=0]+E[Y_{-1}|W=1]=0$.
Under this DGP,
$
E[Y_1(0)-Y_0(0)|W=0] = E[E[Y_1(0)-Y_0(0)|Y_{-1},W=0]|W=0] = E[-Y_{-1}|W=0]
=
E[Y_{-1}|W=1] = E[E[Y_1(0)-Y_0(0)|Y_{-1},W=1]|W=1] = E[Y_1(0)-Y_0(0)|W=1],
$
but
$E[Y_1(0)-Y_0(0)|Y_{-1},W=0] = -Y_{-1} \neq Y_{-1} = E[Y_1(0)-Y_0(0)|Y_{-1},W=1].$
In other words, Condition DID holds but Condition DIDM fails. Intuitively, even if the overall (unconditional) trends are equal (due to cancellation when aggregating over \(Y_{-1}\)), the trends might differ for each subpopulation defined by \(Y_{-1}\).

\begin{comment}
Conversely, consider a setting where untreated outcome trends differ by gender in opposite directions between treated and control groups. For example, suppose that for males
\[
E[Y_1(0)-Y_0(0)\mid \text{male},W=1] = \delta_{\text{male}}^{(1)} \quad \text{and} \quad E[Y_1(0)-Y_0(0)\mid \text{male},W=0] = \delta_{\text{male}}^{(0)},
\]
with \(\delta_{\text{male}}^{(1)} < \delta_{\text{male}}^{(0)}\), and for females
\[
E[Y_1(0)-Y_0(0)\mid \text{female},W=1] = \delta_{\text{female}}^{(1)} \quad \text{and} \quad E[Y_1(0)-Y_0(0)\mid \text{female},W=0] = \delta_{\text{female}}^{(0)},
\]
with \(\delta_{\text{female}}^{(1)} > \delta_{\text{female}}^{(0)}\). Now, if the proportions of males and females in the treated and control groups are such that the weighted averages of these gender-specific trends coincide, then one obtains
\[
\E[Y_1(0)-Y_0(0)\mid W=1] = \E[Y_1(0)-Y_0(0)\mid W=0],
\]
so that the unconditional parallel trends (DID) hold by cancellation. However, within each gender the trends differ between treated and control groups, meaning the conditional parallel trends assumption is violated:
\[
E[Y_1(0)-Y_0(0)\mid \text{gender},W=1] \neq E[Y_1(0)-Y_0(0)\mid \text{gender},W=0].
\]
Thus, while the DID estimand may fortuitously recover the ATT at the aggregate level, the DIDM estimand—which requires valid conditional comparisons—would be biased.
\end{comment}
\bigskip

\noindent\textbf{Implication for Practice.}\\
These examples underscore that each assumption addresses a different facet of the DGP. Practitioners must \textit{ex ante} commit to one of these assumptions based on the context and the nature of available data. Whether one relies on matching (Condition M), unconditional parallel trends (DID), or conditional parallel trends (DIDM) is not a matter of nested robustness but of fundamentally different identifying restrictions, each carrying its own trade-offs and implications for causal inference.

}

%%%%%%%%%%%%%%%%%%%%%%%%%%%%%%%%%%%%%%%%%%%%%%%%%%%%
\subsection{The Lagged Outcome $Y_{-s}$ as A Key Matching Criterion}\label{sec:lagged_outcome}
%%%%%%%%%%%%%%%%%%%%%%%%%%%%%%%%%%%%%%%%%%%%%%%%%%%%
Following \citet{heckman1998characterizing} and \citet{smith2005does}, we place particular emphasis on lagged outcomes, like $Y_{-s}$, as the conditioning variable for M and DIDM.
This focus is motivated by several key studies in the literature.

Firstly, \citet{ashenfelter1978estimating} argues that participants in job training programs typically have permanently lower earnings than non-participants, but also experience a temporary decrease in earnings just before entering the program, a phenomenon now famously known as the ``Ashenfelter dip.'' 
Controlling for lagged outcomes is crucial to addressing this confounding factor, as also emphasized by \citet{heckman1999economics} among others.

In \citet{lalonde1986evaluating} and \citet{dehejia1999causal,dehejia2002propensity}, indeed, including lagged outcomes in matching allowed for close replication of the baseline experimental estimates.
The subsequent discourse between \citet{smith2005does} and \citet{dehejia2005practical} further underscored the significance of conditioning on lagged outcomes.

\citet[][Section 5]{angrist2009mostly} discuss the importance of conditioning on lagged outcomes in the context of lagged-dependent-variable (LDV) models.
More recently, \citet{roth2023parallel}, in their review of difference-in-differences methods, raise an open question about the validity of conditioning on lagged outcomes and whether it reduces or exacerbates bias.

Finally, influential empirical work has emphasized the special role that lagged outcomes play in reducing bias in policy evaluation. For example, \citet{chetty2014measuring1}, in their work on value-added models in education, highlight the utility of lagged outcomes, specifically prior test scores, as essential covariates for obtaining unbiased value-added estimates.

These studies suggest the importance of focusing on lagged outcomes as a key conditioning variable and analyzing the sources of bias that may arise under different conditions.
We thus follow this literature to focus on the lagged outcome $Y_{-s}$ for matching in the current section.
With this said, Section \ref{sec:general} will generalize this setup to include a general class of matching criteria in addition to lagged outcomes.

{\color{black}
Finally, we briefly discuss the lag order, denoted as $-s$. 
In the special case where $s=0$, the DIDM estimand $\attdidm$ simplifies to the M estimand $\attm$. 
As a result, the double bracketing relation $\attm \leq \attdidm \leq \attdid$ that we present later reduces to $\attm = \attdidm \leq \attdid$, which corresponds to the existing bracketing result, LDV $\leq$ FE, as demonstrated by \cite[][Section 5]{angrist2009mostly} and \cite{ding2019bracketing}.
Our framework, which accommodates $s \geq 0$, is more general and particularly includes the symmetric DID approach introduced by \citet{heckman1998characterizing} and \citet{smith2005does}. 
This approach, which aligns with our DIDM with $-s=-1$, was shown by \citet{heckman1998characterizing} and subsequent studies to be more accurate than other specifications. 
Therefore, our arbitrary lag order $-s$ not only encompasses the classical bracketing scenario as a special case but also includes many important cases explored in foundational papers within the literature.
}

%%%%%%%%%%%%%%%%%%%%%%%%%%%%%%%%%%%%%%%%%%%%%%%%%%%%
\section{Empirics of the Double Bracketing}\label{sec:empirical_double_bracketing}
%%%%%%%%%%%%%%%%%%%%%%%%%%%%%%%%%%%%%%%%%%%%%%%%%%%%

This section revisits three empirical studies using four datasets from the literature, parts of which are derived from the seminal papers discussed in Section \ref{sec:introduction}. We are going to document the robustness of the double bracketing relationship, $\attm \leq \attdidm \leq \attdid$, across each of these scenarios, different data sets, estimation methods, and subpopulations defined by observed characteristics. The first two examples, presented in Section \ref{sec:job_training}, focus on job training programs, namely the NSW and JTPA. The final example, in Section \ref{sec:educ}, examines educational programs.

%%%%%%%%%%%%%%%%%%%%%%%%%%%%%%%%%%%%%%%
\subsection{Job Training Programs}\label{sec:job_training}
%%%%%%%%%%%%%%%%%%%%%%%%%%%%%%%%%%%%%%%
\subsubsection{The NSW Program}\label{sec:nsw}
%%%%%%%%%%%%%%%%%%%%%%%%%%%%%%%%%%%%%%%

We begin our analysis with a reexamination of the National Supported Work (NSW) programs \citep{lalonde1986evaluating,dehejia1999causal,dehejia2002propensity,smith2005does} through the lens of our double bracketing. We employ the same sets of the CPS and PSID data as those utilized by \citet{lalonde1986evaluating} and \citet{smith2005does}. Although these data sets have been extensively analyzed in the literature, we offer a brief description in Appendix \ref{sec:appendix:nsw} for the sake of completeness and convenience of the readers.

The primary outcome variable, $Y_t$, represents participants' self-reported earnings, adjusted to 1982 dollars. The treatment variable, $W$, is a binary indicator denoting whether an individual was assigned to the NSW program. Additionally, demographic variables, including race and educational level, are included as auxiliary covariates in all M, DIDM, and DID estimations.

Let us revisit Figure \ref{fig:chabe_signed}, which was initially introduced in Section \ref{sec:introduction} to motivate our investigation. Recall that Figure \ref{fig:chabe_signed} illustrates the signed \textit{biases} of the estimates for the three alternative estimands, $\attm$, $\attdidm$, and $\attdid$, relative to the experimental estimates in percentage terms according to \citet[HIST; ][]{heckman1998characterizing} and \citet[ST; ][]{smith2005does}. 

In the current section, we focus on the black bars in Figure \ref{fig:chabe_signed}, representing the estimates by \citet[ST; ][]{smith2005does}. Notably, the double bracketing inequality, $\attm \leq \attdidm \leq \attdid$, is upheld when considering their point estimates. Since these represent the \textit{biases}, note that $\attm$ and $\attdidm$ are biased downward, whereas $\attdid$ is biased upward. These results suggest that the M estimand, $\attm$, tends to be conservative, while the DID estimand, $\attdid$, appears to be optimistic, assuming that the experimental estimates are the true values.

The observation above applies specifically to the estimates of $\attm$, $\attdidm$, and $\attdid$ selected by \cite{chabe2017should} from the various ST specifications. However, when we extend our analysis to include iterative estimations across all nine estimation methods employed by \citet{smith2005does} using the CPS and PSID datasets, we find that the double bracketing relationship, $\attm \leq \attdidm \leq \attdid$, is robustly upheld across the vast majority of these specifications. This is illustrated in Figure \ref{fig:smith_todd_all}, where the top and bottom charts display the results for the CPS and PSID data, respectively.
%%%%%%%%%%%%%%%%%%%%%%%%%%%%%%%%%%%%%%%
\begin{figure}[t]
\centering
\includegraphics[width=0.6\textwidth]{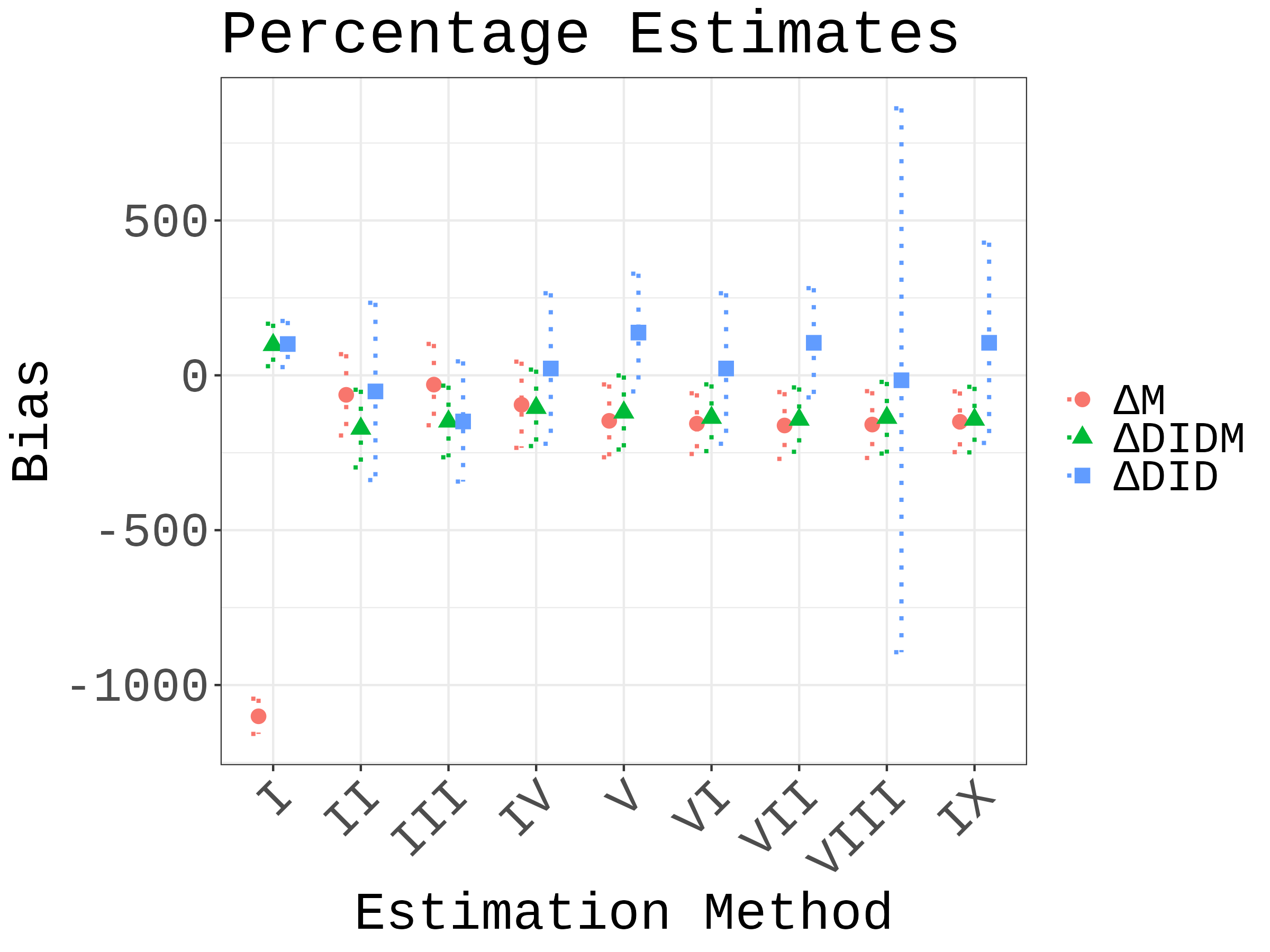}
\\${}$\\
\includegraphics[width=0.6\textwidth]{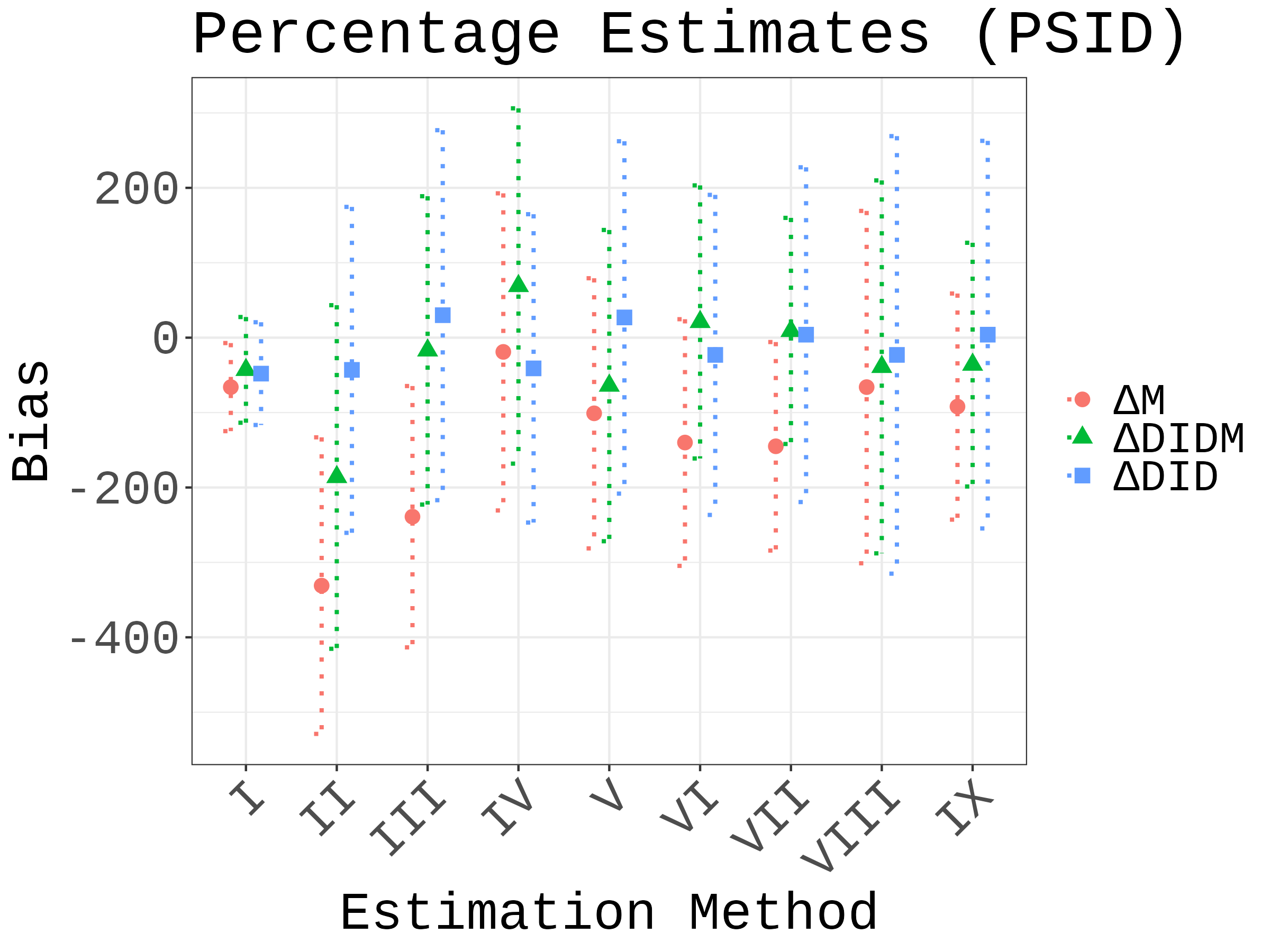}
\caption{The signed biases of M, DIDM, and DID estimates across all the estimation methods in Tables 5--6 of \citet{smith2005does} with the CPS data (top) and the PSID data (bottom). The estimation methods are labeled as follows: I - Mean difference, II - 1 Nearest neighbor without support, III - 10 Nearest-neighbors without support, IV - 1 Nearest-neighbor with support, V - 10 Nearest-neighbors with support, VI - Local linear matching (bw = 1.0), VII - Local linear matching (bw = 4.0), VIII - Local linear regression adjusted (bw = 1.0), IX - Local linear regression adjusted (bw = 4.0).}${}$
\label{fig:smith_todd_all}
\end{figure}
%%%%%%%%%%%%%%%%%%%%%%%%%%%%%%%%%%%%%%%
Moreover, it is important to note that this ordering of estimates is never rejected for any of the nine specifications. This consistency ensures that the result is not an artifact of a specific estimation procedure.
Consistently, we observe that the M estimand, $\attm$, tends to be conservative, while the DID estimand, $\attdid$, tends to be optimistic.

%%%%%%%%%%%%%%%%%%%%%%%%%%%%%%%%%%%
\subsubsection{The JTPA Program}\label{sec:jtpa}
%%%%%%%%%%%%%%%%%%%%%%%%%%%%%%%%%%%

Next, we reexamine the Job Training Partnership Act (JTPA) program studied by \cite{heckman1998characterizing}. Although the dataset they used has been extensively analyzed in the literature, a brief description is provided in Appendix \ref{sec:appendix:jtpa} for the sake of completeness and convenience of the readers.

The outcome variable, $Y_t$, represents participants' earnings adjusted for inflation. The treatment variable, $W$, indicates whether an individual was assigned to the JTPA program. Additionally, demographic covariates such as sex and age are included in all M, DIDM, and DID estimations.

Once again, we reiterate Figure \ref{fig:chabe_signed} that we provide in Section \ref{sec:introduction}. 
For the current section about the JTPA programs, focus on the gray bars in Figure \ref{fig:chabe_signed} based on the estimates by \citet[HIST; ][]{heckman1998characterizing}.
Notice that the inequality relationship, $\attm \leq \attdidm \leq \attdid$, is maintained when considering their point estimates. Specifically, the M estimand, $\attm$, is biased downward, while the DID estimand, $\attdid$, is biased upward.

%%%%%%%%%%%%%%%%%%%%%%%%%%%%%%%%%%%%%%%
\subsection{Educational Programs}\label{sec:educ}
%%%%%%%%%%%%%%%%%%%%%%%%%%%%%%%%%%%%%%%

Now, we turn to the educational programs analyzed by \citet{athey2020combining}, which were also studied in \cite{chetty2014measuring1, chetty2014measuring2}. In their study, \citet{athey2020combining} addressed the challenge of selection in observational by combining it with short-term experimental data. In contrast, our analysis focuses solely on observational data to examine the behavior of the observational estimands: $\attm$, $\attdidm$, and $\attdid$. 
Details of the data are provided in Appendix \ref{sec:appendix:educ}.

The outcome variable, $Y_t$, represents students' test scores, specifically standardized scores that average results from mathematics and English language arts. The treatment variable, $W$, indicates whether a student was assigned to a small class size. Additionally, covariates such as gender, race, and eligibility for free lunch are considered to define subpopulations in our analysis.

Given the large size of our dataset (see Appendix \ref{sec:appendix:educ} for details), we can precisely estimate the M, DIDM, and DID estimands even when the sample is divided into subpopulations based on observed attributes. Figure \ref{fig:educ:bracket} presents these estimates, along with their 95\% confidence intervals, for each subpopulation.
%%%%%%%%%%%%%%%%%%%%%%%%%%%%%%%%%%%%%%%
\begin{figure}[t]
\centering
\includegraphics[width=0.7\textwidth]{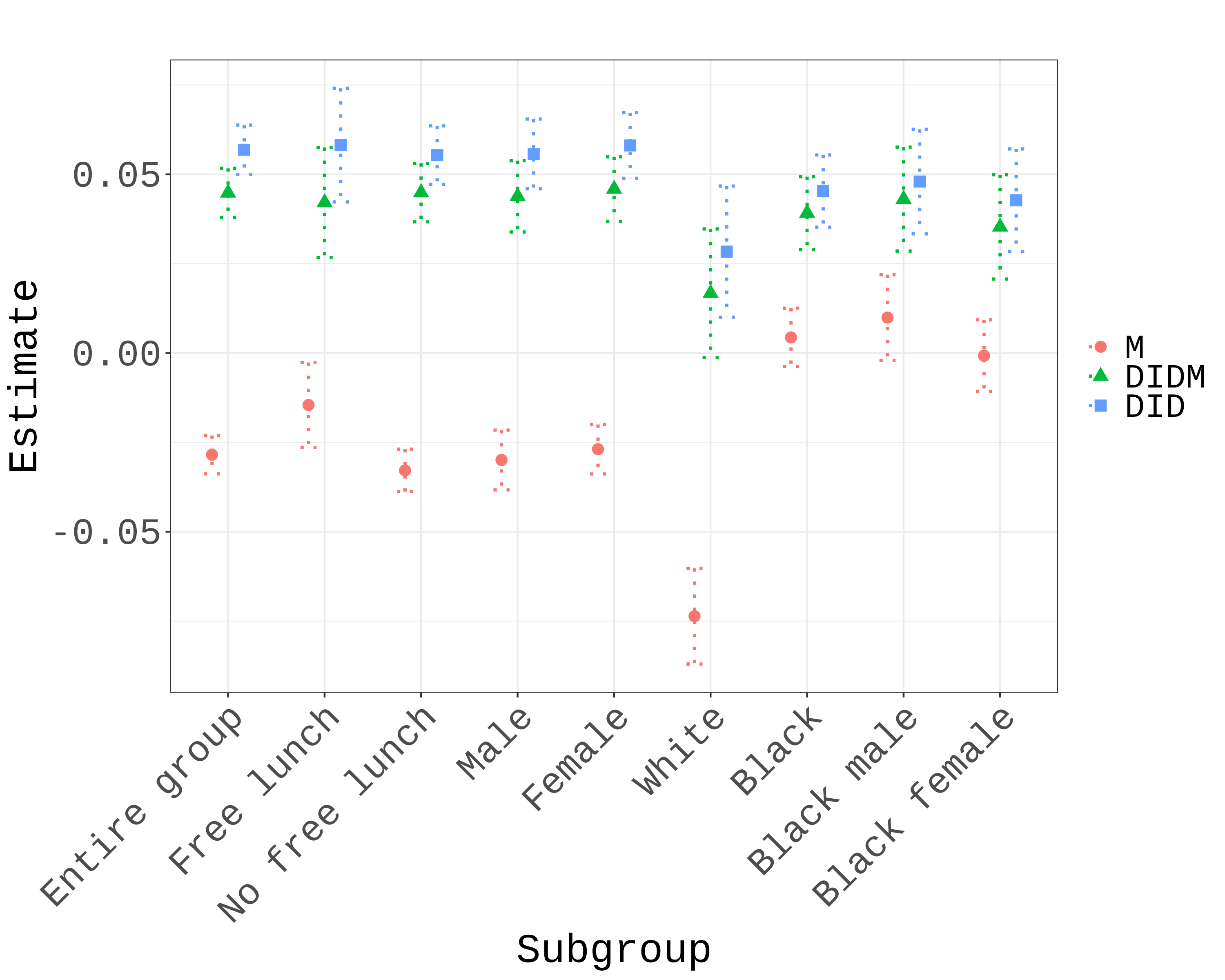}
\caption{Estimates along with their 95\% confidence intervals of the M, DIDM, and DID estimands for the effects of small class sizes on test scores for each subpopulation defined by observed attributes.}${}$
\label{fig:educ:bracket}
\end{figure}
%%%%%%%%%%%%%%%%%%%%%%%%%%%%%%%%%%%%%%%

Observe that the double bracketing relationship, $\attm \leq \attdidm \leq \attdid$, remains consistent across various subpopulations, despite notable differences in levels among these groups. Interestingly, the M estimand, $\attm$, generally indicates negative effects, except within the Black subpopulations. In contrast, both the DIDM estimand, $\attdidm$, and the DID estimand, $\attdid$, typically suggest positive effects. This divergence in implications between estimation methods raises questions about relying solely on observational data, highlighting the value of combining experimental and observational data, as demonstrated by \citet{athey2020combining}.

\section{The Double Bracketing: A Parametric Approach}\label{sec:parametric}
%%%%%%%%%%%%%%%%%%%%%%%%%%%%%%%%%%%%%%%%%%%%%%%%%%%%

So far, we have observed that the double bracketing relationship, $\attm \leq \attdidm \leq \attdid$, appears to be more of a rule than a mere coincidence in empirical studies on educational and job training programs. 
In this section, and the two subsequent ones, we are going to provide theoretical explanations for this intriguing phenomenon.

While general theories will be presented in the subsequent sections, we begin by illustrating how the double bracketing relationship,
$
\attm \leq \attdidm \leq \attdid
$
arises under a familiar parametric setting.
In particular, suppose the true data-generating process takes the form:
\begin{equation}\label{eq:parametric}
Y_{i,t} = \alpha_i + \beta W_i \mathbbm{1}\{t \geq 1\} + \gamma W_i + \delta_t  + \rho Y_{i,t-1} + \epsilon_{i,t},
\end{equation}
where the error term satisfies 
\begin{equation}\label{eq:parametric_orthogonality}
E[\epsilon_{i,1}|Y_{i,-1},W_i]=E[\epsilon_{i,0}|Y_{i,-1},W_i]=0.
\end{equation}
In this context, $\alpha_i$ and $\delta_t$ serve as two-way fixed effects and $\rho$ serves as a persistence parameter.

Note that $\gamma$ can be expressed as
$
\gamma = E[Y_{i,0}|W_i=1,Y_{i,-1}] - E[Y_{i,0}|W_i=0,Y_{i,-1}]
$
under \eqref{eq:parametric}--\eqref{eq:parametric_orthogonality}.
This value, $\gamma$, is non-positive if individuals with weakly lower $Y_{i,0}$ tend to select into treatment $W_i=1$ conditional on $Y_{i,-1}$ as is plausibly the case with job training programs.
We postulate this negative selection assumption:
\begin{equation}\label{eq:parametric:negative_selection_1}
\text{Negative Selection I:} \qquad \gamma = E[Y_{i,0}|W_i=1,Y_{i,-1}] - E[Y_{i,0}|W_i=0,Y_{i,-1}] \leq 0.
\end{equation}
We also postulate a similar negative selection assumption for $Y_{i,-1}$:
\begin{equation}\label{eq:parametric:negative_selection_2}
\text{Negative Selection II:} \qquad E[Y_{i,-1}|W_i=1] \leq E[Y_{i-1}|W_i=0].
\end{equation}
Finally, we impose the assumption of a non-explosive earnings process:
\begin{equation}\label{eq:parametric:non_explosive}
\text{Non-Explosive Process:} \qquad 0 \leq \rho \leq 1.
\end{equation}
Note that these three conditions \eqref{eq:parametric:negative_selection_1}--\eqref{eq:parametric:non_explosive} are analogous to the assumptions invoked by \citet[][Section 5]{angrist2009mostly} in developing their bracketing relationship, LDV $\leq$ FE.

Under the parametric data-generating process \eqref{eq:parametric}, some calculations yield
\begin{align}
\attm &= \beta + (1+\rho)\gamma,
\label{eq:parametric:m}
\\
\attdidm
&= \beta + \rho\gamma,
\qquad\qquad\text{and}
\label{eq:parametric:didm}
\\
\attdid &= \beta + \rho\gamma + \rho(1-\rho) (E[Y_{i,-1}|W_i=0] - E[Y_{i,-1}|W_i=1]).
\label{eq:parametric:did}
\end{align}
See Appendix \ref{sec:detailed_calculations} for detailed calculations that derive these expressions.
From \eqref{eq:parametric:m}--\eqref{eq:parametric:didm}, we have
\begin{equation}\label{eq:parametric:didm_m}
\attdidm-\attm = -\gamma \geq 0,
\end{equation}
where the inequality follows from $\gamma \leq 0$ due to \eqref{eq:parametric:negative_selection_1}.
From \eqref{eq:parametric:didm}--\eqref{eq:parametric:did}, we have
\begin{equation}\label{eq:parametric:did_didm}
\attdid-\attdidm = \rho(1-\rho) (E[Y_{i,-1}|W_i=0] - E[Y_{i,-1}|W_i=1]) \geq 0,
\end{equation}
where the inequality follows from $E[Y_{i,-1}|W_i=0] - E[Y_{i,-1}|W_i=1] \geq 0$ due to \eqref{eq:parametric:negative_selection_2} and $\rho \in [0,1]$ due to \eqref{eq:parametric:non_explosive}.

Combining \eqref{eq:parametric:didm_m}--\eqref{eq:parametric:did_didm} yields the double bracketing relationship,
\begin{align*}
\attm \leq \attdidm \leq \attdid,
\end{align*}
which holds true under the simple parametric model \eqref{eq:parametric} with the plausible restrictions \eqref{eq:parametric:negative_selection_1}--\eqref{eq:parametric:non_explosive} motivated by \citet[][Section 5]{angrist2009mostly}.

Finally, we highlight a couple of special cases in which the three estimands collapse into two.
In the absence of the first negative selection (i.e., $\gamma = 0$) the double bracketing relationship reduces to
\begin{align*}
\attm = \attdidm \leq \attdid.
\end{align*}
In other words, M and DIDM become equivalent.
Similarly, in the absence of persistence (i.d., $\rho=0$) or under the unit root (i.e., $\rho=1$), the double bracketing relationship reduces to
\begin{align*}
\attm \leq \attdidm = \attdid.
\end{align*}
In other words, DIDM and DID become equivalent.
In these two special cases, DIDM indeed becomes redundant.
However, it is generally distinct from the other two estimands otherwise.

%%%%%%%%%%%%%%%%%%%%%%%%%%%%%%%%%%%%%%%%%%%%%%%%%%%%
\section{The Double Bracketing in Nonparametric Frameworks}\label{sec:simple}
%%%%%%%%%%%%%%%%%%%%%%%%%%%%%%%%%%%%%%%%%%%%%%%%%%%%

The previous section derived the double-bracketing relationship $\attm \leq \attdidm \leq \attdid$ under a simple parametric model framework.
To ensure generality, we adopt a non-parametric and model-free framework in the current section.

Let
\begin{align*}
\Delta(\attm) =& \attm-\att,
\\
\Delta(\attdid) =& \attdid-\att, \qquad\text{and}
\\
\Delta(\attdidm) =& \attdidm-\att
\end{align*}
be the identification errors of the estimands, $\attm$, $\attdid$, and $\attdidm$, respectively.
Observe that $\attm \leq \attdidm \leq \attdid$ is equivalent to $\Delta(\attm) \leq \Delta(\attdidm) \leq \Delta(\attdid)$ with these notations.
Thus, we are going to establish the double bracketing relation, 
$\Delta(\attm) \leq \Delta(\attdidm) \leq \Delta(\attdid)$, in terms of the identification error.
To this end, we consider the following assumption.

\begin{assumption}\label{a:special}The following conditions hold.
\begin{enumerate}[(i)]
\item\label{a:special:selection}
$E[Y_0|W=0,Y_{-s}=y] \geq E[Y_0|W=1,Y_{-s}=y]$ for all $y$.
\item \label{a:special:sd}
$F_{Y_{-s}|W=0}$
first-order stochastically dominates
$F_{Y_{-s}|W=1}$.
\item\label{a:special:dec}
$y \mapsto \Phi(y) := E[Y_1-Y_0|W=0,Y_{-s}=y]$ is weakly decreasing.
\end{enumerate}
\end{assumption}

{\color{black}
The three components, \eqref{a:special:selection}--\eqref{a:special:dec}, of Assumption \ref{a:special} are analogous to \eqref{eq:parametric:negative_selection_1}--\eqref{eq:parametric:non_explosive}, respectively.
Each of the three components is empirically testable, as it does not involve unobserved latent variables such as $Y_t(d)$.
Furthermore, they are plausible and align with the assumptions made in the literature \citep[e.g.,][]{angrist2009mostly,ding2019bracketing}.
Specifically, part \eqref{a:special:selection} requires the so-called ``negative selection'' that individuals who opt out from treatment (i.e., those with $W=0$) tend to have weakly higher pre-treatment (potential) outcomes $Y_0=Y_0(0)$ without treatment on average given $Y_{-s}=y$.
Similarly, part \eqref{a:special:sd} requires the negative selection that individuals who opt out from treatment (i.e., those with $W=0$) tend to have no lower pre-treatment (potential) outcome $Y_{-s}=Y_{-s}(0)$ without treatment.
Part \eqref{a:special:dec} requires the time series $\{Y_t\}_t$ of outcomes for those with $W=0$ (i.e., $\{Y_t\}_t = \{Y_t(0)\}_t$) is stable over time. In real-world settings, especially among populations with limited socioeconomic status (SES), it may be reasonable to expect that outcomes do not escalate dramatically without intervention. For example, in educational programs targeting low-SES individuals, we wouldn't anticipate rapid, unsustainable improvements in outcomes without structured support on average.
 In the special case where $s=0$, this requirement essentially reflects the condition that the root of the AR(1) process is less than one, which is equivalent to stationarity, as deployed in, e.g., \cite{chetty2014measuring1}.

Each of these parts is analogous to the assumption made by \citet[][pp. 185]{angrist2009mostly} in deriving their bracketing relationship between the DID and lagged dependent variable (LDV) in linear and additive models.
In the special case of $s=0$, part \eqref{a:special:selection} will be trivially satisfied while parts \eqref{a:special:sd}--\eqref{a:special:dec} correspond to the assumptions invoked by \citet{ding2019bracketing}.
As argued at the end of Section \ref{sec:lagged_outcome}, our framework with $s \geq 0$ accommodates both this special case and other important scenarios considered by \citet{heckman1998characterizing} and others.
}

\begin{theorem}\label{theorem:main}
Suppose that Assumption \ref{a:special} holds.
Then, we have the bracketing relationship
\begin{align*}
\Delta(\attm) \leq \Delta(\attdidm) \leq \Delta(\attdid).
\end{align*}
\end{theorem}

\begin{proof}[Proof of Theorem \ref{theorem:main}]
Note that the identification errors can be written as
\begin{align*}
\Delta\left(\attm\right)= & E\left[E\left[Y_1(0) \mid W=1, Y_{-s}\right]-E\left[Y_1(0) \mid W=0, Y_{-s}\right] \mid W=1\right],
\\
\Delta\left(\attdid\right)= & E\left[Y_1(0)-Y_0(0) \mid W=1\right]-E\left[Y_1(0)-Y_0(0) \mid W=0\right],
\qquad\text{and}
\\
\Delta\left(\attdidm\right)= & E\left[E\left[Y_1(0)-Y_0(0) \mid W=1, Y_{-s}\right]-E\left[Y_1(0)-Y_0(0) \mid W=0, Y_{-s}\right] \mid W=1\right].
\end{align*}

First, observe that
\begin{align*}
\Delta(\attdidm) - \Delta(\attm)
=&
E[E[Y_0(0)|W=0,Y_{-s}] 
- E[Y_0(0)|W=1,Y_{-s}]|W=1]
\\
=&
E[E[Y_0|W=0,Y_{-s}] 
- E[Y_0|W=1,Y_{-s}]|W=1]
\geq 0,
\end{align*}
where 
the second equality is due to $Y_t(0)=Y_t$ for all $t \leq 0$, and
the last inequality follows from Assumption \ref{a:special} \eqref{a:special:selection}.

Second, observe that
\begin{align*}
&\Delta(\attdid) - \Delta(\attdidm)
\\
=&
E[E[Y_1(0)-Y_0(0)|W=0,Y_{-s}] | W=1]
-
E[Y_1(0)-Y_0(0)|W=0]
\\
=&
E[E[Y_1(0)-Y_0(0)|W=0,Y_{-s}] | W=1]
-
E[E[Y_1(0)-Y_0(0)|W=0,Y_{-s}] | W=0]
\\
=&
E[E[Y_1-Y_0|W=0,Y_{-s}] | W=1]
-
E[E[Y_1-Y_0|W=0,Y_{-s}] | W=0]
\\
=& 
E[\Phi(Y_{-s}) | W=1]
-
E[\Phi(Y_{-s}) | W=0]
\geq 0,
\end{align*}
where 
the second equality follows from the law of iterated expectations,
the third equality is due to $Y_t(0)=Y_t$ given $W=0$, and
the last inequality follows from Assumption \ref{a:special} \eqref{a:special:sd}--\eqref{a:special:dec}.
\end{proof}

This theorem provides a theoretical explanation for the double bracketing relationship, $\attm \leq \attdidm \leq \attdid$, which was robustly observed in Section \ref{sec:empirical_double_bracketing}.

Furthermore, this theorem implies the following three consequences depending on the underlying DGP.
First, when Condition M is true, then
\begin{align*}
0 = \Delta(\attm) \leq \Delta(\attdidm) \leq \Delta(\attdid).
\end{align*}
In this case, $\attm$ identifies $\att$ but $\attdidm$ and $\attdid$ tend to be upwardly biased.
Second, when Condition DID is true, then
\begin{align*}
\Delta(\attm) \leq \Delta(\attdidm) \leq \Delta(\attdid) = 0.
\end{align*}
In this case, $\attdid$ identifies $\att$ but $\attm$ and $\att$ tend to be downwardly biased.
Finally, when Condition DIDM is true, then
\begin{align*}
\Delta(\attm) \leq \Delta(\attdidm) = 0 \leq \Delta(\attdid).
\end{align*}
In this case, $\attdidm$ identifies $\att$ but $\attm$ tends to be downwardly biased while $\attdid$ tends to be upwardly biased.

Hence, for applications where non-negative treatment effects are expected, the DID estimand will generally incur optimistic estimates while the M estimand will produce conservative estimates.
The DIDM estimand can be optimistic or conservative, depending on the underlying DGP.
In summary, the M estimand is conservatively the most robust.

%%%%%%%%%%%%%%%%%%%%%%%%%%%%%%%%%%%%%%%%%%%%%%%%%%%%
\section{Extension to General Cases}\label{sec:general}
%%%%%%%%%%%%%%%%%%%%%%%%%%%%%%%%%%%%%%%%%%%%%%%%%%%%

The main theoretical result presented in Section \ref{sec:simple} extends to a more general class of data-generating processes and associated estimands.
In this section, we present an extension to general cases with a focus on the recent developments in the methods of event studies as principal examples.

%%%%%%%%%%%%%%%%%%%%%%%%%%%%%%%%%%%%%%%%%%%%%%%%%%%%
\subsection{Setup}
%%%%%%%%%%%%%%%%%%%%%%%%%%%%%%%%%%%%%%%%%%%%%%%%%%%%

The previous notations do \textit{not} carry over to the current section. 
Suppose that a researcher is interested in identifying the average treatment effect on the treated (ATT) defined by 
\begin{align}
\att = E\left[\left. \widetilde Y_1(1) - \widetilde Y_1(0) \right|W=1\right].
\label{eq:general:att}
\end{align}
At this moment, we have not introduced the specific meanings of the notations.
They will be discussed in the contexts of specific examples in Section \ref{sec:examples}.
With this said, we want to remark that they parallel with those notations introduced in Section \ref{sec:simple}.
Unlike the previous section, however, the subscripts no longer indicate the time in general.

Similarly to the previous section, we define the alternative estimands
\begin{align}
\attm &= E\left[\left. \widetilde Y_1 \right| W=1\right] - E\left[\left. E\left[\left. \widetilde Y_1 \right|W=0, X \right] \right| W=1\right],
\label{eq:general:m}
\\
\attdid &= E\left[\left. \widetilde Y_1 - \widetilde Y_0 \right| W= 1\right] - E\left[\left. \widetilde Y_1 - \widetilde Y_0 \right| W=0\right],
\qquad\text{and}
\label{eq:general:did}
\\
\attdidm &= E\left[\left. E\left[\left. \widetilde Y_1 - \widetilde Y_0 \right| X, W=1\right] - E\left[\left. \widetilde Y_1 - \widetilde Y_0 \right| X, W=0\right] \right|W=1\right],
\label{eq:general:didm}
\end{align}
called the matching (M), the difference-in-differences (DID), and the difference-in-differences matching (DIDM), respectively.
The $p$-dimensional random vector $X$ is now used as a matching criterion.

The following conditions are imposed:
\begin{align}\label{eq:pretreatment}
\text{Pre-Treatment:}&&    \widetilde Y_0(0)=&\widetilde Y_0 \text{ a.s. given $W \in \{0,1\}$.}
\\
\label{eq:comparison}
\text{Comparison:}&&    \widetilde Y_1(0)=&\widetilde Y_1 \text{ a.s. given } W=0.
\end{align}
Condition \eqref{eq:pretreatment} requires that observed outcomes be the potential outcome without treatment for every unit prior to treatment.
Condition \eqref{eq:comparison} requires that the observed outcome be the potential outcome without treatment for the control group.

\subsection{Examples of the M, DID, and DIDM Estimands in Event Studies}\label{sec:examples}

In this section, we demonstrate that our general framework \eqref{eq:general:att}--\eqref{eq:comparison} encompasses alternative estimands studied in the literature of event studies as examples.
%%%%%%%%%%%%%%%%%%%%%%%%%%%%%%%%%%%%%%%%%%%%%%%%%%%%
\subsubsection{Example 1: M in Event Studies}\label{sec:event_m}
%%%%%%%%%%%%%%%%%%%%%%%%%%%%%%%%%%%%%%%%%%%%%%%%%%%%

\citet{acemoglu2019democracy} consider the ATT
\begin{equation}\label{eq:acemoglu:att}
E\left[ Y_t^s(1) - Y_t^s(0) | D_t=1, D_{t-1}=0\right],
\end{equation}
where $D_t$ denotes the indicator of democracy, $Y_t^s(1)$ denotes the potential GDP in period $t+s$ when a country is treated between periods $t-1$ and $t$ (i.e., $D_t=1$ and $D_{t-1}=0$), and $Y_t^s(0)$ denotes the potential GDP in period $t+s$ when such a treatment does not occur (i.e., $D_t=D_{t-1}=0$).\footnote{The original paper by \citet{acemoglu2019democracy} considers
$
E\left[ (Y_t^s(1)-Y_{t-1}) - (Y_t^s(0)-Y_{t-1}) | D_t=1, D_{t-1}=0\right]
$
as the parameter of interest, 
where $Y_{t-1}$ denotes the realized GDP at period $t$,
but this is equivalent to \eqref{eq:acemoglu:att}.}
\citet{acemoglu2019democracy} identify this ATT by
\begin{equation}\label{eq:acemoglu:estimand_pre}
E\left[\left. Y_t^s-Y_{t-1} \right| D_t=1, D_{t-1}=0 \right]
-
E\left[\left. E\left[\left. Y_t^s-Y_{t-1} \right| D_t=0, D_{t-1}=0, X \right] \right| D_t=1, D_{t-1} = 0 \right],
\end{equation}
where
$Y_t^s$ denotes the observed GDP in period $t+s$,
$Y_{t-1}$ denotes the observed GDP in period $t-1$, and
$X := (Y_{t-1},\ldots,Y_{t-4})'$
in their baseline model with additional covariates in extended robustness analyses.

Since $X$ contains $Y_{t-1}$ in particular, the conditioning theorem\footnote{Specifically, the conditioning theorem yields $E[Y_{t-1}|D_t=0,D_{t-1}=0,X]=Y_{t-1}$ when $X$ contains $Y_{t-1}$.} cancels $Y_{t-1}$ between the two terms in \eqref{eq:acemoglu:estimand_pre}, so the identifying formula \eqref{eq:acemoglu:estimand_pre} of \citet{acemoglu2019democracy} boils down to
\begin{equation}\label{eq:acemoglu:estimand}
E\left[\left. Y_t^s \right| D_t=1, D_{t-1}=0 \right]
-
E\left[\left. E\left[\left. Y_t^s \right| D_t=0, D_{t-1}=0, X \right] \right| D_t=1, D_{t-1} = 0 \right].
\end{equation}

Our general framework encompasses this example.
Specifically, our ATT \eqref{eq:general:att} reduces to \eqref{eq:acemoglu:att} and our M estimand \eqref{eq:general:m} reduces to \eqref{eq:acemoglu:estimand} by setting
\begin{align*}
\widetilde Y_0(0) :=& Y_{t-1},
&
\widetilde Y_0(1) :=& Y_{t-1},
&
\widetilde Y_0 :=& Y_{t-1},
\\
\widetilde Y_1(0) :=& Y_t^s(0),
&
\widetilde Y_1(1) :=& Y_t^s(1),
&
\widetilde Y_1 :=& Y_t^s,
\end{align*}
\begin{align*}
\text{and} \qquad
W := 
\begin{cases}
1 & \text{if } D_t=1 \text{ and } D_{t-1}=0
\\
0 & \text{if } D_t=0 \text{ and } D_{t-1}=0
\\
-1 & \text{otherwise}
\end{cases}
\end{align*}
The pre-treatment condition \eqref{eq:pretreatment} is satisfied by construction, as $\widetilde Y_0(0) = Y_{t-1} = \widetilde Y_0$.
The comparison condition \eqref{eq:comparison} is also satisfied by construction via the definition of $Y_t^s(0)$ as the potential outcome under $D_t=D_{t-1}=0$.
Namely, $\widetilde Y_1(0) = Y_t^s(0) = Y_t^s = \widetilde Y_1$ holds given $D_t=D_{t-1}=0$.

\subsubsection{Example 2: DID in Event Studies}\label{sec:event_did}
%%%%%%%%%%%%%%%%%%%%%%%%%%%%%%%%%%%%%%%%%%%%%%%%%%%%
\citet{callaway2018difference} consider the ATT
\begin{equation}\label{eq:callaway:att}
E\left[\left. Y_t(g) - Y_t(\infty) \right| G=g \right],
\end{equation}
where $G$ denotes the treatment period, $Y_t(g)$ denotes the potential outcome at period $t \geq g$ when an individual is treated at period $g$, and $Y_t(\infty)$ denotes the potential outcome at period $t$ when an individual does not receive a treatment.
\citet{callaway2018difference} identify this ATT by
\begin{equation}\label{eq:callaway:estimand}
E\left[\left.Y_t - Y_{g-1} \right| G=g \right]
-
E\left[\left.Y_t - Y_{g-1} \right| G=g' \right]
\end{equation}
for $g' \geq t+1$,
where $Y_t$ denotes the observed outcome at period $t$.
(The second term of \eqref{eq:callaway:estimand} may be aggregated over $G' \in \{t+1,t+2,\ldots\}$.)

Our general framework encompasses this example.
Specifically, our ATT \eqref{eq:general:att} reduces to \eqref{eq:callaway:att} and our DID estimand \eqref{eq:general:did} reduces to \eqref{eq:callaway:estimand} by setting
\begin{align*}
\widetilde Y_0(0) :=& Y_{g-1}(\infty),
&
\widetilde Y_0(1) :=& Y_{g-1}(g),
&
\widetilde Y_0 :=& Y_{g-1},
\\
\widetilde Y_1(0) :=& Y_t(\infty),
&
\widetilde Y_1(1) :=& Y_t(g),
&
\widetilde Y_1 :=& Y_t,
\end{align*}
\begin{align*}
\text{and} \qquad
W := 
\begin{cases}
1 & \text{if } G=g
\\
0 & \text{if } G=g'
\\
-1 & \text{otherwise}
\end{cases}
\end{align*}
The pre-treatment condition \eqref{eq:pretreatment} is satisfied by construction, as $\widetilde Y_0(0) = Y_{g-1}(\infty) = Y_{g-1} = \widetilde Y_0$ given $G=g$ or $G=g' \geq t+1 > g$.
The comparison condition \eqref{eq:comparison} is also satisfied by construction, as $\widetilde Y_1(0) = Y_t(\infty) = Y_t = \widetilde Y_1$ given $G=g' \geq t+1$.

%%%%%%%%%%%%%%%%%%%%%%%%%%%%%%%%%%%%%%%%%%%%%%%%%%%%
\subsubsection{Example 3: DIDM in Event Studies}\label{sec:event_m}
%%%%%%%%%%%%%%%%%%%%%%%%%%%%%%%%%%%%%%%%%%%%%%%%%%%%
\citet[][Section 4.1]{dube2023local} consider the ATT
\begin{equation}\label{eq:dube:att}
E\left[ Y_{t+h}(1) - Y_{t+h}(0) | \Delta D_t=1\right],
\end{equation}
where $\Delta D_t$ denotes the indicator of policy change, $Y_{t+h}(1)$ denotes the potential outcome in period $t+h$ when a policy changes between periods $t-1$ and $t$ (i.e., $\Delta _t=1$), and $\Delta Y_{t+h}(0)$ denotes the potential outcome in period $t+h$ when such a change does not occur (i.e., $\Delta D_t=0$).
\citet[][Section 4.1]{dube2023local} identify this ATT by
\begin{equation}\label{eq:dube:estimand}
E\left[\left. Y_{t+h}-Y_{t-1} \right| \Delta D_t=1 \right]
-
E\left[\left. E\left[\left. Y_{t+h}-Y_{t-1} \right| \Delta D_t=0, X \right] \right| \Delta D_t=1 \right],
\end{equation}
where
$Y_{t+h}$ denotes the observed outcome in period $t+h$,
$Y_{t-1}$ denotes the observed outcome in period $t-1$, and
$X$ is a vector of general covariates.

Our general framework encompasses this example.
Specifically, our ATT \eqref{eq:general:att} reduces to \eqref{eq:dube:att} and our DIDM estimand \eqref{eq:general:didm} reduces to \eqref{eq:dube:estimand} by setting
\begin{align*}
\widetilde Y_0(0) :=& Y_{t-1},
&
\widetilde Y_0(1) :=& Y_{t-1},
&
\widetilde Y_0 :=& Y_{t-1},
\\
\widetilde Y_1(0) :=& Y_{t+h}(0),
&
\widetilde Y_1(1) :=& Y_{t+h}(1),
&
\widetilde Y_1 :=& Y_{t+h},
\end{align*}
\begin{align*}
\text{and} \qquad
W := 
\begin{cases}
1 & \text{if } \Delta D_t=1
\\
0 & \text{if } \Delta D_t=0
\\
-1 & \text{otherwise}
\end{cases}
\end{align*}
The pre-treatment condition \eqref{eq:pretreatment} is satisfied by construction, as $\widetilde Y_0(0) = Y_{t-1} = \widetilde Y_0$.
The comparison condition \eqref{eq:comparison} is also satisfied by construction via the definition of $Y_{t+h}(0)$ as the potential outcome under $\Delta D_t=0$.
Namely, $\widetilde Y_1(0) = Y_{t+h}(0) = Y_{t+h} = \widetilde Y_1$ holds given $\Delta D_t=0$.

Also see the DID${}_\text{M}$ estimator of \citet{deChaisemartin2020two}, and the (panel) matching estimator of \citet{imai2023matching}, as well as the extended DID method of \citet[][Section 4.1]{dube2023local} -- they all propose and analyze the properties of what we refer to as the DIDM.

As pointed out by \citet{dube2023local}, their framework encompasses \citet{acemoglu2019democracy} as a special case. Indeed, when $X$ contains $\widetilde Y_0$, as is the case with \citet{acemoglu2019democracy} presented in Section \ref{sec:event_m}, our DIDM framework reduces to our M framework.
In general, however, the DIDM differs from the M.

\subsubsection{Summary and Discussions of the Three Examples}
%%%%%%%%%%%%%%%%%%%%%%%%%%%%%%%%%%%%%%%%%%%%%%%%%%%%
Albeit there are slight differences in their notations, the three examples presented above focus on similar setups.
They fundamentally differ only in terms of the estimands: the three examples focus on the M, DID, and DIDM estimands in our language.
In general, a researcher does not know which of them achieves the identification.
The M estimand identifies the true ATT (i.e, $\attm = \att$ holds) if the matching condition
\begin{align*}
\text{Condition M:} \ \ \ \left.\left(\widetilde Y_1(1), \widetilde Y_1(0) \right) \indep W \right| X
\end{align*}
is satisfied.
The DID estimand identifies the true ATT (i.e, $\attdid = \att$ holds) if the parallel trend condition
\begin{align*}
\text{Condition DID:} \ \ \
\E\left[\left.\widetilde Y_1(0)-\widetilde Y_0(0)\right|W=0\right]
=
\E\left[\left.\widetilde Y_1(0)-\widetilde Y_0(0)\right|W=1\right]
\end{align*}
is satisfied.
Finally, the DIDM estimand identifies the true ATT (i.e, $\attdidm = \att$ holds) if the conditional parallel trend condition
\begin{align*}
\text{Condition DIDM:} \ \ \
\E\left[\left.\widetilde Y_1(0)-\widetilde Y_0(0)\right|X,W=0\right]
=
\E\left[\left.\widetilde Y_1(0)-\widetilde Y_0(0)\right|X,W=1\right]
\end{align*}
is satisfied.
In the absence of knowledge of the underlying data-generating process, however, committing to a wrong assumption can lead to biased estimates by M, DID, or DIDM.
It is therefore of interest to characterize the relation among the three estimands.
The following subsection investigates this point.

%%%%%%%%%%%%%%%%%%%%%%%%%%%%%%%%%%%%%%%%%%%%%%%%%%%%
\subsection{The General Double Bracketing Result}
%%%%%%%%%%%%%%%%%%%%%%%%%%%%%%%%%%%%%%%%%%%%%%%%%%%%
Now, focus on the generic framework \eqref{eq:general:att}--\eqref{eq:comparison} again.
Let
\begin{align*}
\Delta(\attm) =& \attm-\att,
\\
\Delta(\attdid) =& \attdid-\att, \qquad\text{and}
\\
\Delta(\attdidm) =& \attdidm-\att
\end{align*}
be the identification errors of the estimands, $\attm$, $\attdid$, and $\attdidm$, respectively.
We establish the double bracketing relation 
$\Delta(\attm) \leq \Delta(\attdidm) \leq \Delta(\attdid)$ under the following assumption.

\begin{assumption}\label{a:general}The following conditions hold.
\begin{enumerate}[(i)]
\item\label{a:general:selection}
$E\left[\left.\widetilde Y_0 \right|W=0,X=x\right] \geq E\left[\left. \widetilde Y_0 \right|W=1,X=x\right]$ for all $x$.
\item \label{a:general:sd}
$F_{X|W=0}$
multivariate first-order stochastically dominates
$F_{X|W=1}$.\footnote{We say that $X$ multivariate first-order stochastically dominates $X^\ast$ if $P(X \leq x) \leq P(X^\ast \leq x)$ holds for all $x \in \mathbb{R}^p$.}
\item\label{a:general:dec}
$x \mapsto \Phi(x) := E\left[\left.\widetilde Y_1-\widetilde Y_0\right|W=0,X=x\right]$ is weakly decreasing.\footnote{We say that $\Phi$ is weakly decreasing if $\Phi(x_1,\ldots,x_p) \geq \Phi(x^\ast_1,\ldots,x^\ast_p)$ holds whenever $x_1 \leq x^\ast_1$, $\ldots$ and, $x_p \leq x^\ast_p$.}
\end{enumerate}
\end{assumption}
The three parts \eqref{a:general:selection}--\eqref{a:general:dec} of this assumption parallel those in Assumption \ref{a:special}, albeit that $X$ is now possibly multi-dimensional.
Hence, similar interpretations can be made especially when $X$ consists of lagged outcomes as in the first example \textit{a la} \citet{acemoglu2019democracy} presented in Section \ref{sec:event_m}.
Such a convenient interpretation may not be feasible if $X$ contains other covariates, but we want to stress that each of the three conditions \eqref{a:general:selection}--\eqref{a:general:dec} of this assumption is still empirically testable.

The following theorem states the extended double bracketing result for the general cases.

\begin{theorem}\label{theorem:general}
Suppose that Assumption \ref{a:general} holds for \eqref{eq:general:att}--\eqref{eq:comparison}.
Then, we have the bracketing relationship
\begin{align*}
\Delta(\attm) \leq \Delta(\attdidm) \leq \Delta(\attdid).
\end{align*}
\end{theorem}

\begin{proof}[Proof of Theorem \ref{theorem:general}  ]
Note that the identification errors can be written as
\begin{align*}
\Delta\left(\attm\right)= & E\left[E\left[\widetilde Y_1(0) \mid W=1, X\right]-E\left[\widetilde Y_1(0) \mid W=0, X\right] \mid W=1\right],
\\
\Delta\left(\attdid\right)= & E\left[\widetilde Y_1(0)-\widetilde Y_0(0) \mid W=1\right]-E\left[\widetilde Y_1(0)-\widetilde Y_0(0) \mid W=0\right],
\qquad\text{and}
\\
\Delta\left(\attdidm\right)= & E\left[E\left[\widetilde Y_1(0)-\widetilde Y_0(0) \mid W=1, X\right]-E\left[\widetilde Y_1(0)-\widetilde Y_0(0) \mid W=0, X\right] \mid W=1\right].
\end{align*}

First, observe that
\begin{align*}
\Delta(\attdidm) - \Delta(\attm)
=&
E\left[E\left[\widetilde Y_0(0)|W=0,X\right] 
- E\left[\widetilde Y_0(0)|W=1,X\right]|W=1\right]
\\
=&
E\left[E\left[\widetilde Y_0|W=0,X\right] 
- E\left[\widetilde Y_0|W=1,\right]|W=1\right]
\geq 0
\end{align*}
where 
the second equality is due to \eqref{eq:pretreatment}, and
the last inequality follows from Assumption \ref{a:general} \eqref{a:general:selection}.

Second, observe that
\begin{align*}
&\Delta(\attdid) - \Delta(\attdidm)
\\
=&
E\left[E\left[\widetilde Y_1(0)-\widetilde Y_0(0)|W=0,X\right] | W=1\right]
-
E\left[\widetilde Y_1(0)-\widetilde Y_0(0)|W=0\right]
\\
=&
E\left[E\left[\widetilde Y_1(0)-\widetilde Y_0(0)|W=0,X\right] | W=1\right]
-
E\left[E\left[\widetilde Y_1(0)-\widetilde Y_0(0)|W=0,X\right] | W=0\right]
\\
=&
E\left[E\left[\widetilde Y_1-\widetilde Y_0|W=0,X\right] | W=1\right]
-
E\left[E\left[\widetilde Y_1-\widetilde Y_0|W=0,X\right] | W=0\right]
\\
=& 
E\left[\Phi(X) | W=1\right]
-
E\left[\Phi(X) | W=0\right]
\geq 0,
\end{align*}
where 
the second equality follows from the law of iterated expectations,
the third equality is due to \eqref{eq:pretreatment}--\eqref{eq:comparison}, and
the last inequality follows from Assumption \ref{a:general} \eqref{a:general:sd}--\eqref{a:general:dec}.
\end{proof}

%%%%%%%%%%%%%%%%%%%%%%%%%%%%%%%%%%%%%%%
\section{Empirical Evidence of the Assumptions}\label{sec:empirical_assumption}
%%%%%%%%%%%%%%%%%%%%%%%%%%%%%%%%%%%%%%%

With Sections \ref{sec:simple}--\ref{sec:general} providing theoretical justifications for the double bracketing relationship empirically observed in Section \ref{sec:empirical_double_bracketing}, we will now examine whether the underlying assumption (Assumption \ref{a:special} or \ref{a:general}) of our theory is satisfied by the datasets used in Section \ref{sec:empirical_double_bracketing}.
Specifically, we revisit the CPS and PSID datasets from Section \ref{sec:nsw} and the observational dataset from Section \ref{sec:educ}.\footnote{Unfortunately, we are unable to present our analysis for the JTPA program discussed in Section \ref{sec:jtpa} due to discrepancies between the currently available microdata and the original microdata used by the authors, which we confirmed through repeated communications with them.}

%%%%%%%%%%%%%%%%%%%%%%%%%%%%%%%%%%%%%%%
\subsection{Job Training Programs}\label{sec:nsw:assumption}
%%%%%%%%%%%%%%%%%%%%%%%%%%%%%%%%%%%%%%%

Recall that Section \ref{sec:nsw} demonstrates the robustness of the double bracketing relationship, $\attm \leq \attdidm \le \attdid$, for the NSW program using both the CPS data set and the PSID data set.
In light of these consistent observations and our theoretical prediction provided in Theorem \ref{theorem:main}, the current section examines each of the three parts, (i)--(iii), of Assumption \ref{a:special} in detail, using both data sets.

The first condition is Assumption \ref{a:special} \eqref{a:special:selection} requiring that $E[Y_0|W=0,Y_{-s}=y] \geq E[Y_0|W=1,Y_{-s}=y]$ holds for all $y$.
To check this condition, we estimate the conditional expectation functions, $y \mapsto E[Y_0|W=0,Y_{-s}=y]$ and $y \mapsto E[Y_0|W=1,Y_{-s}=y]$, non-parametrically by the partitioning-based least squares regression.\footnote{\label{foot:nonparametric_estimation}We use the R package \texttt{lspartition} developed by \citet{cattaneo2019lspartition}. We took all the default parameters.}
Their estimates are plotted in Figure \ref{fig:nsw:ass1}.
%%%%%%%%%%%%%%%%%%%%%%%%%%%%%%%%%%%%%%%
\begin{figure}[t]
\centering
CPS\\
\includegraphics[width=0.55\textwidth]{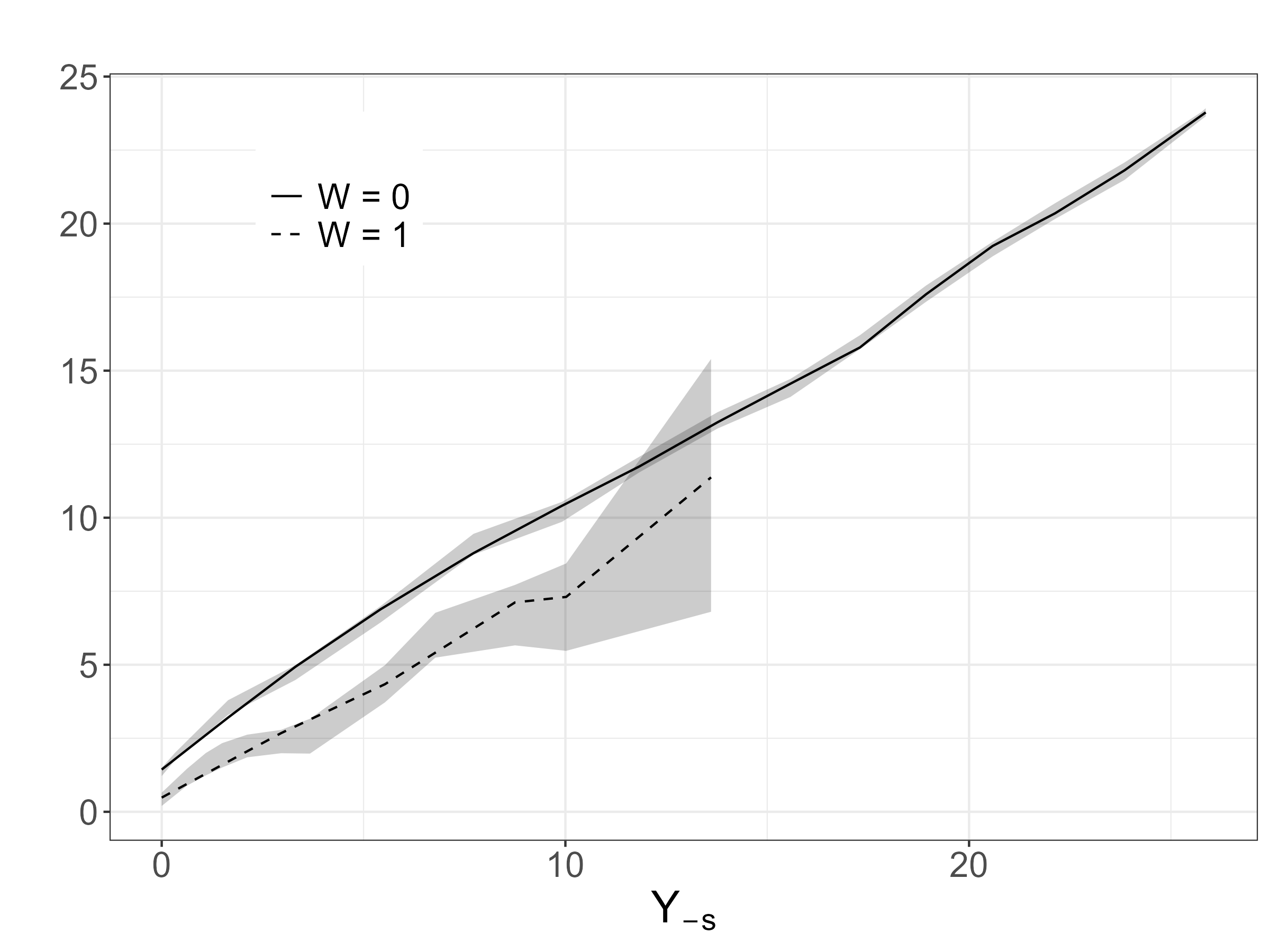}
\\
PSID\\
\includegraphics[width=0.55\textwidth]{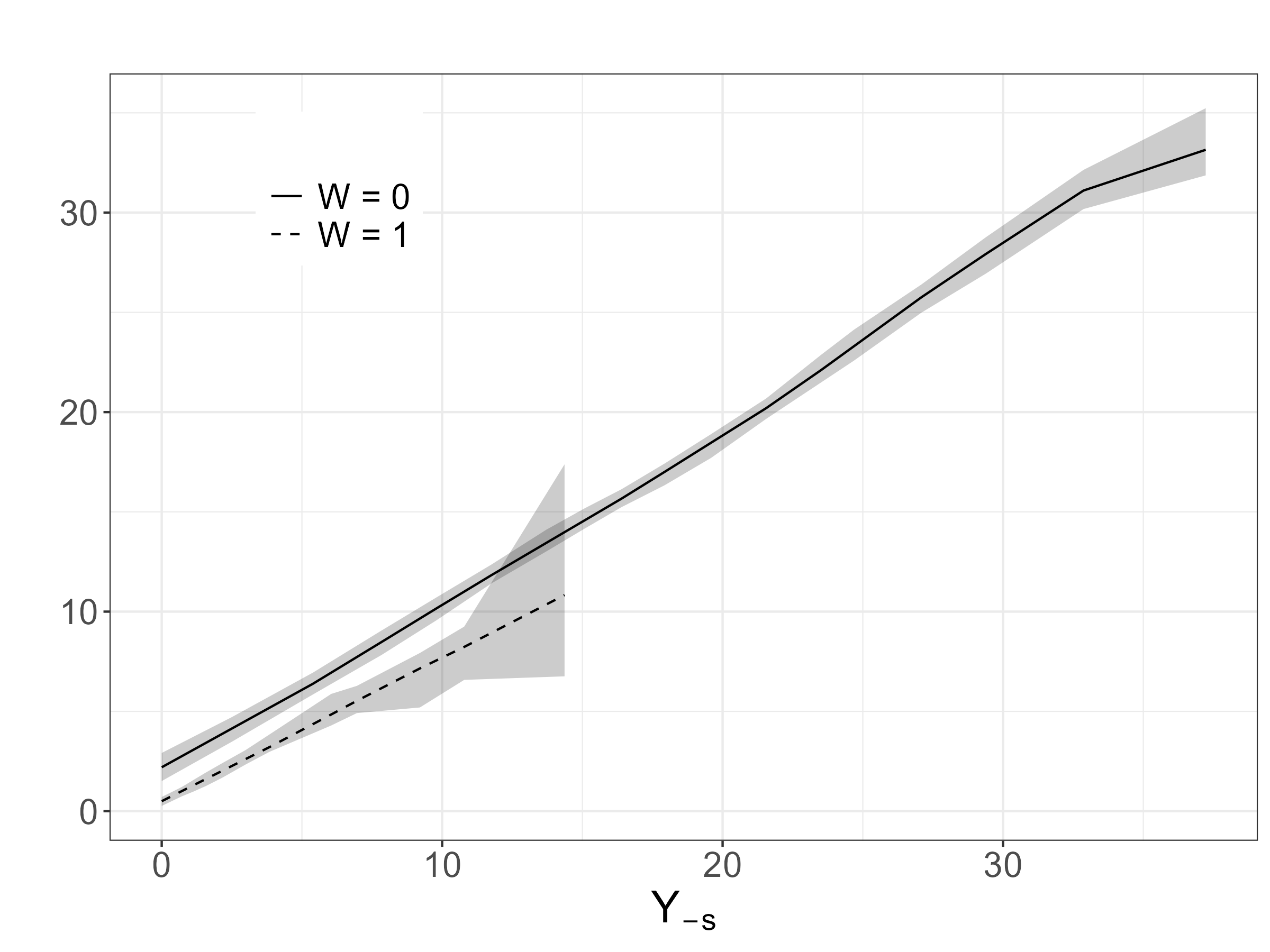}
\caption{Evidence of Assumption \ref{a:special} \eqref{a:special:selection} for the NSW program using the CPS (top) and PSID (bottom) data sets. The solid and dashed lines represent estimates of the conditional expectation functions $y \mapsto E[Y_0|W=0,Y_{-s}=y]$ and $y \mapsto E[Y_0|W=1,Y_{-s}=y]$, respectively. Shaded areas denote 95\% confidence bands. Both axes are measured in thousands of U.S. dollars. For details on the estimation method, refer to Footnote \ref{foot:nonparametric_estimation}.}${}$
\label{fig:nsw:ass1}
\end{figure}
%%%%%%%%%%%%%%%%%%%%%%%%%%%%%%%%%%%%%%%
The solid and dashed lines represent estimates of the conditional expectation functions $y \mapsto E[Y_0|W=0,Y_{-s}=y]$ and $y \mapsto E[Y_0|W=1,Y_{-s}=y]$, respectively. 
Shaded areas denote their 95\% confidence bands.\footnote{\label{foot:noncentered}We remark that the confidence bands are not centered around the estimates in \texttt{lspartition}, because of bias correction.}
The top image in this figure is based on the CPS data set while the bottom one is based on the PSID data set.
For both of the two data sets, the inequality $E[Y_0|W=0,Y_{-s}=y] \geq E[Y_0|W=1,Y_{-s}=y]$ is clearly satisfied for all $y$, providing evidence in support of our Assumption \ref{a:special} \eqref{a:special:selection}.

The second condition, Assumption \ref{a:special} \eqref{a:special:sd}, requires that $F_{Y_{-s}|W=0}$ first-order stochastically dominates $F_{Y_{-s}|W=1}$.
We estimate the conditional cumulative distribution functions, $F_{Y_{-s}|W=0}$ and $F_{Y_{-s}|W=1}$, by their empirical counterparts, i.e., empirical CDFs conditionally on $W=0$ and $W=1$, respectively.
Their estimates are plotted in Figure \ref{fig:nsw:ass3}.
%%%%%%%%%%%%%%%%%%%%%%%%%%%%%%%%%%%%%%%
\begin{figure}[t]
\centering
\begin{tabular}{cc}
CPS & PSID\\
\includegraphics[width=0.5\textwidth]{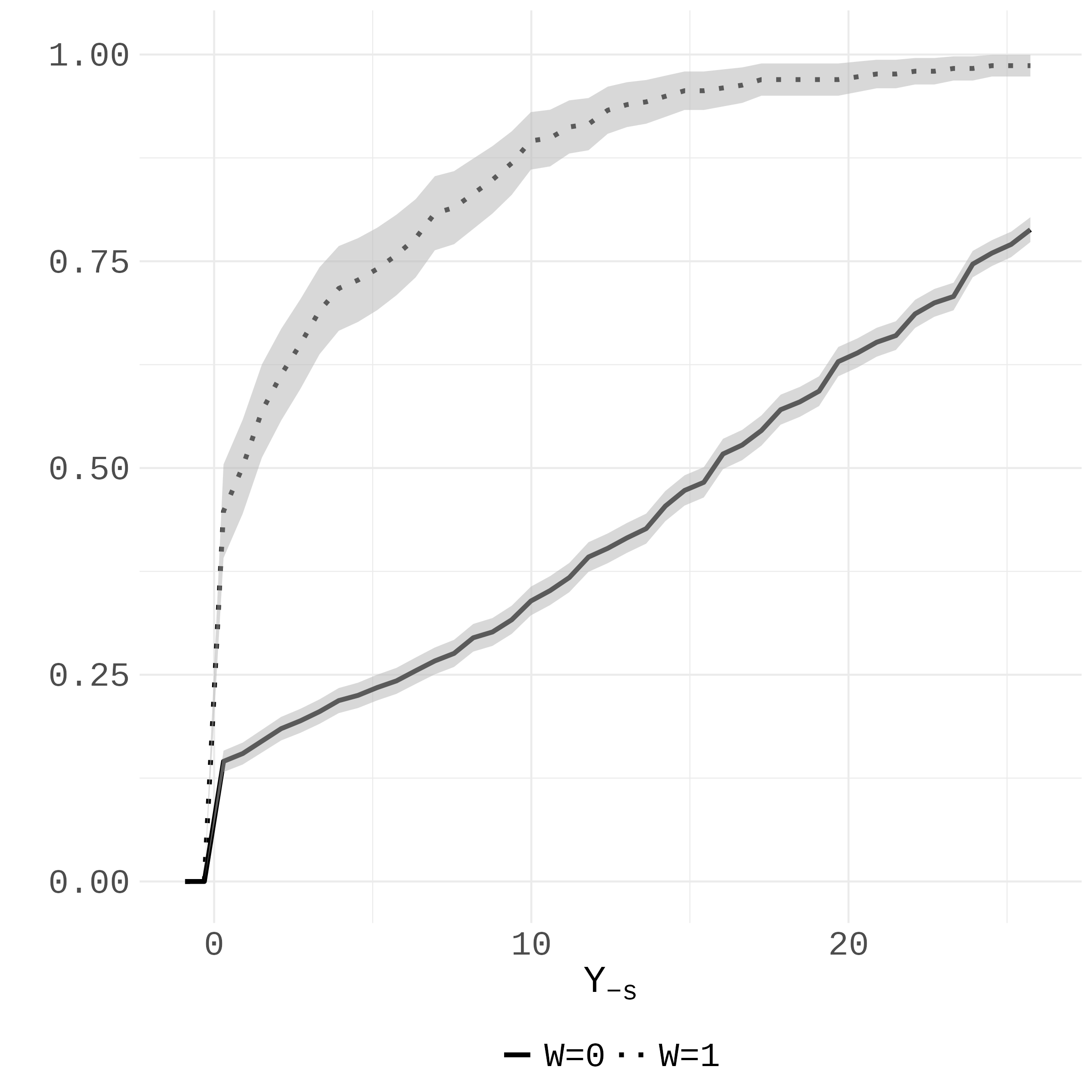}
&
\includegraphics[width=0.5\textwidth]{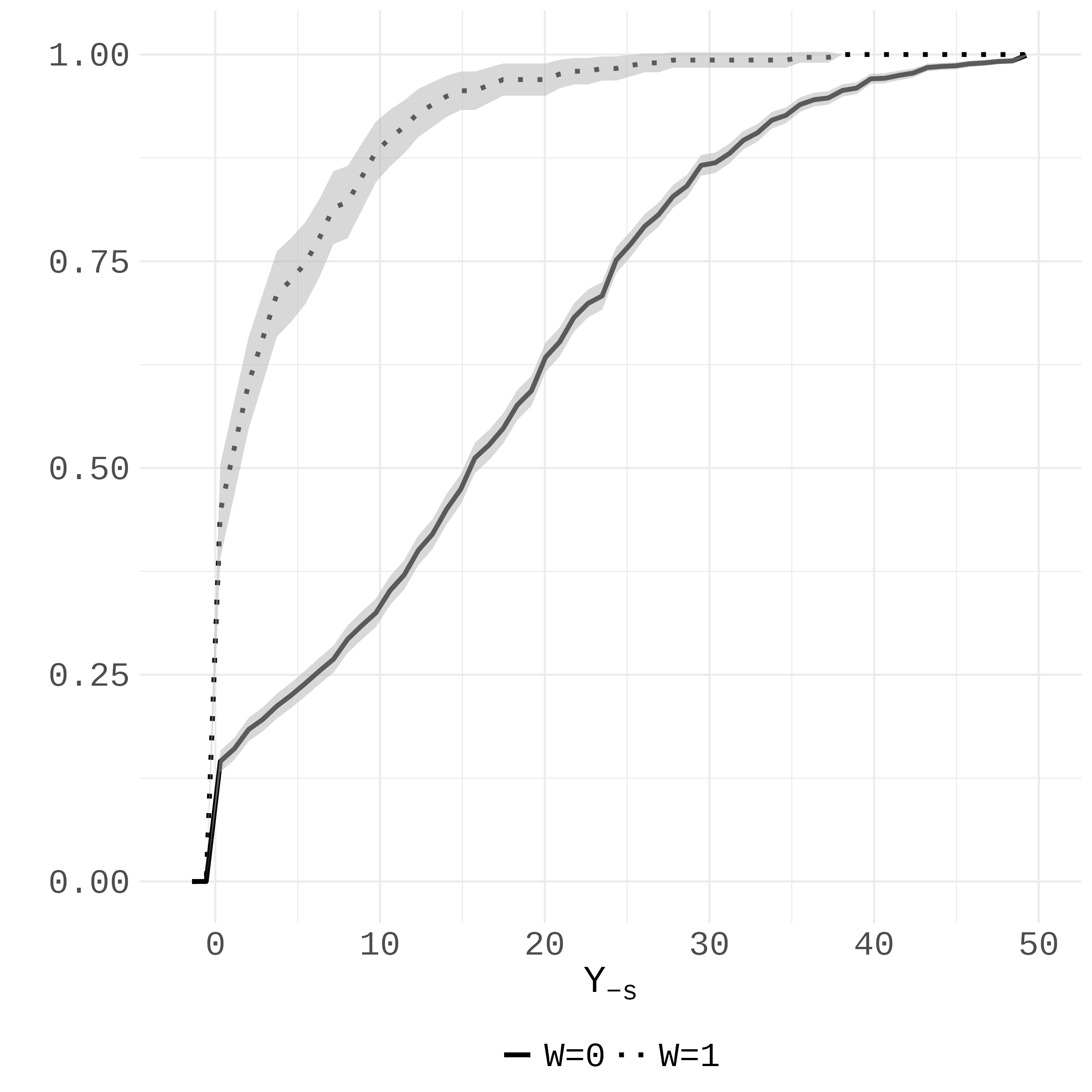}
\end{tabular}
\caption{Evidence of Assumption \ref{a:special} \eqref{a:special:sd} for the NSW program using the CPS (left) and PSID (right) data sets. The conditional cumulative distribution functions, $F_{Y_{-s}|W=0}$ and $F_{Y_{-s}|W=1}$, are estimated by their empirical counterparts, i.e., empirical CDFs. The estimates, along with their 95\% confidence intervals, are plotted. The horizontal axes are measured in thousands of U.S. dollars.}${}$
\label{fig:nsw:ass3}
\end{figure}
%%%%%%%%%%%%%%%%%%%%%%%%%%%%%%%%%%%%%%%
The solid line represents the estimates of $F_{Y_{-s}|W=0}$ while the dotted line represents those of $F_{Y_{-s}|W=1}$.
Their 95\% confidence intervals are indicated by the shaded regions.
The left figure is based on the CPS data set while the right one is based on the PSID data set.
For both of the two data sets, the figure clearly shows that $F_{Y_{-s}|W=0}$ first-order stochastically dominates $F_{Y_{-s}|W=1}$, providing evidence in support of our Assumption \ref{a:special} \eqref{a:special:sd}.

The third condition, Assumption \ref{a:special} \eqref{a:special:dec}, requres the function $y \mapsto \Phi(y) := E[Y_1-Y_0|W=0,Y_{-s}=y]$ to be weakly decreasing.
We estimate the conditional expectation function $\Phi$ non-parametrically by using the partitioning-based least squares regression -- see Footnote \ref{foot:nonparametric_estimation} for details.
The estimates, along with their 95\% confidence bands, are plotted in Figure \ref{fig:nsw:ass2}.
%%%%%%%%%%%%%%%%%%%%%%%%%%%%%%%%%%%%%%%
\begin{figure}[t]
\centering
CPS\\
\includegraphics[width=0.55\textwidth]{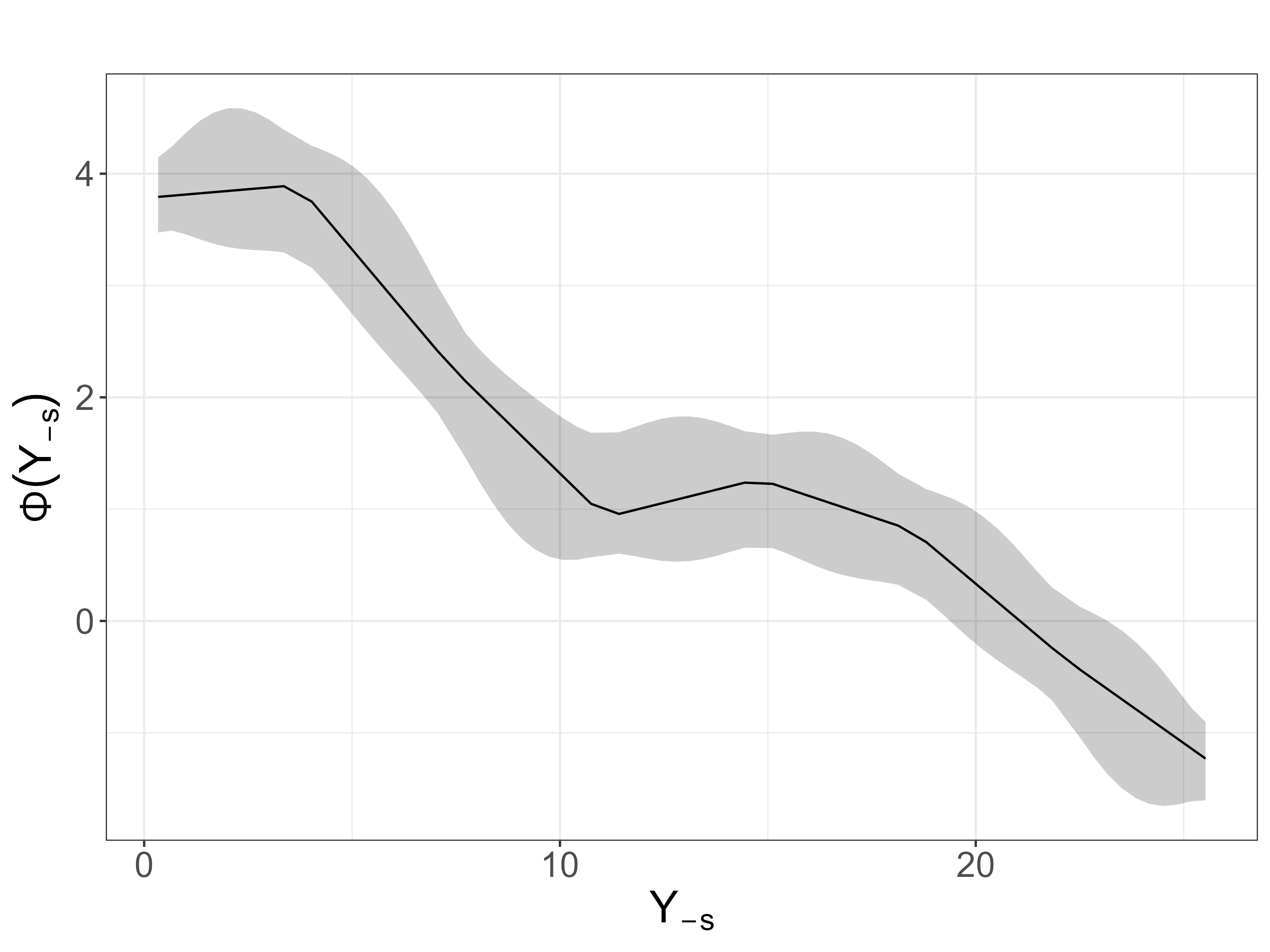}
\\
PSID\\
\includegraphics[width=0.55\textwidth]{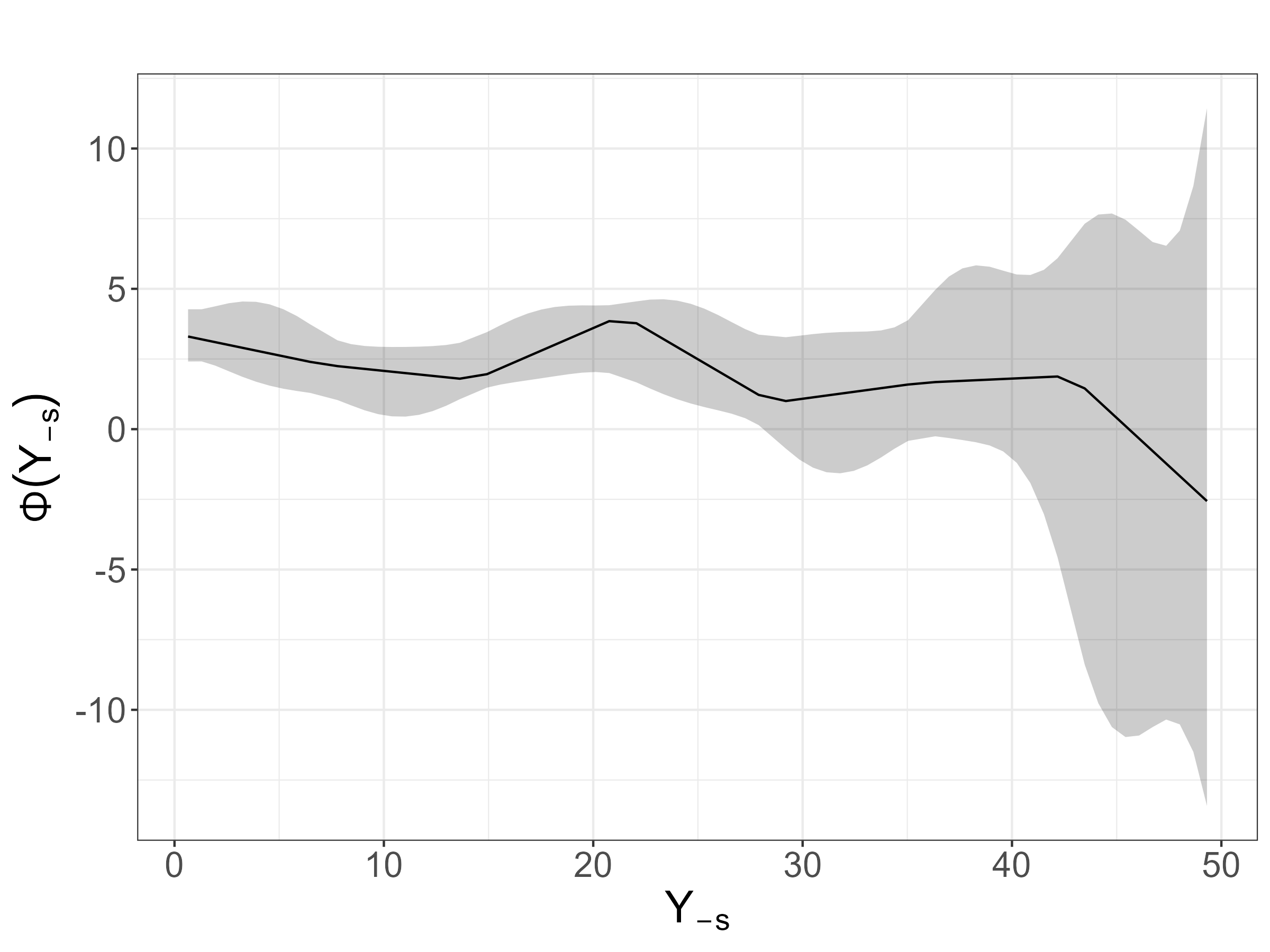}
\caption{Evidence of Assumption \ref{a:special} \eqref{a:special:dec} for the NSW program using the CPS (top) and PSID (bottom) data sets. The conditional expectation function $\Phi$ is estimated non-parametrically by the partitioning-based least squares regression. The estimates, along with their 95\% confidence bands, are plotted.  Both the vertical and horizontal axes are measured in thousands of U.S. dollars. For details on the estimation method, refer to Footnote \ref{foot:nonparametric_estimation}.}${}$
\label{fig:nsw:ass2}
\end{figure}
%%%%%%%%%%%%%%%%%%%%%%%%%%%%%%%%%%%%%%%
The top figure is based on the CPS data set while the bottom one is based on the PSID data set.
The confidence bands imply that the weak decreasingness of this function $\Phi$ cannot be refuted for either of the two data sets. 
Thus, they provide evidence in support of our Assumption \ref{a:special} \eqref{a:special:dec}.

In the current section, we do not include auxiliary covariates in the current analysis.
In Appendix \ref{sec:additional:nsw}, we provide further evidence in support of our assumption even after accounting for auxiliary covariates.

In summary, all three components, (i)--(iii), of Assumption \ref{a:special} hold fairly robustly for both the CPS and PSID data sets regardless of whether auxiliary covariates are included or not.
Hence, the double bracketing relationship $\Delta(\attm) \leq \Delta(\attdidm) \leq \Delta(\attdid)$, as predicted by our Theorem \ref{theorem:main}, is expected to hold, which we have empirically confirmed in Section \ref{sec:nsw}.

\subsection{Educational Programs}\label{sec:educ:assumption}
%%%%%%%%%%%%%%%%%%%%%%%%%%%%%%%%%%%%%%%

Recall that Section \ref{sec:educ} demonstrates that the double bracketing relationship, $\attm \leq \attdidm \leq \attdid$, is robust for the educational program investigated by \citet{athey2020combining}. 
In light of this observation and our theoretical prediction provided in Theorem \ref{theorem:main}, our next question is whether Assumption \ref{a:special} for the double bracketing theory is satisfied for this educational program. To address this, we will examine each of the three parts, (i)--(iii), of Assumption \ref{a:special} in detail using the data set utilized in Section \ref{sec:educ}.

Recall that the first condition is Assumption \ref{a:special} \eqref{a:special:selection}, requiring that $E[Y_0|W=0,Y_{-s}=y] \geq E[Y_0|W=1,Y_{-s}=y]$ holds for all $y$.
We estimate the conditional expectation functions, $y \mapsto E[Y_0|W=0,Y_{-s}=y]$ and $y \mapsto E[Y_0|W=1,Y_{-s}=y]$, non-parametrically using the partitioning-based least squares regression -- see Footnotes \ref{foot:nonparametric_estimation}--\ref{foot:noncentered} for details.
The estimates are plotted in Figure \ref{fig:educ:ass1}.
%%%%%%%%%%%%%%%%%%%%%%%%%%%%%%%%%%%%%%%
\begin{figure}[t]
\centering
\includegraphics[width=0.55\textwidth]{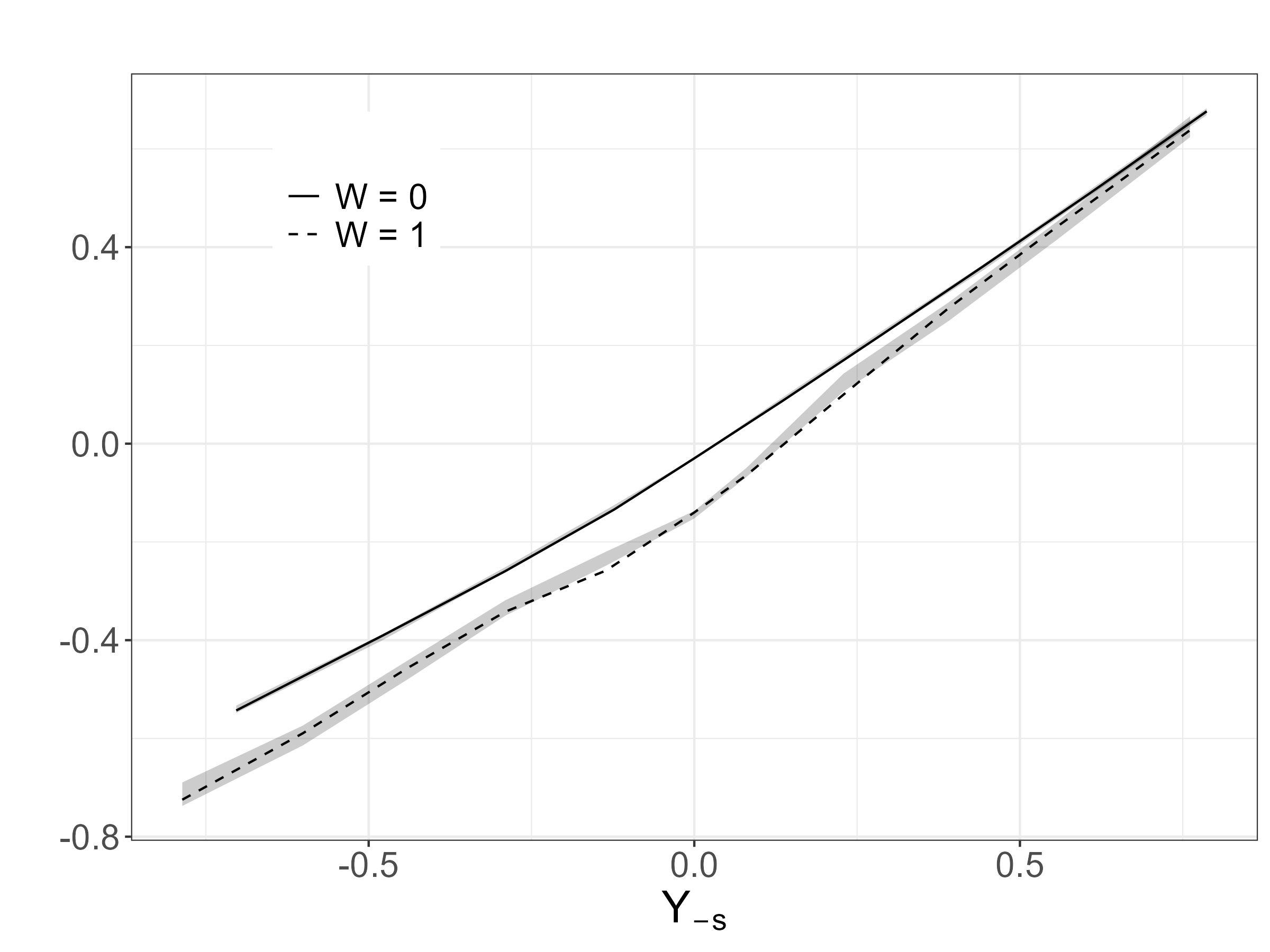}
\caption{Evidence of Assumption \ref{a:special} \eqref{a:special:selection} for the educational program. The solid and dashed lines represent estimates of the conditional expectation functions $y \mapsto E[Y_0|W=0,Y_{-s}=y]$ and $y \mapsto E[Y_0|W=1,Y_{-s}=y]$, respectively. Shaded areas denote 95\% confidence bands. For details on the estimation method, refer to Footnote \ref{foot:nonparametric_estimation}.}${}$
\label{fig:educ:ass1}
\end{figure}
%%%%%%%%%%%%%%%%%%%%%%%%%%%%%%%%%%%%%%%
The solid and dashed lines represent estimates of the conditional expectation functions $y \mapsto E[Y_0|W=0,Y_{-s}=y]$ and $y \mapsto E[Y_0|W=1,Y_{-s}=y]$, respectively.
Shaded areas denote their 95\% confidence bands, although they are nearly invisible due to the large sample size.\footnote{We remark that the estimates appear outside of the confidence bands because the bands are centered around bias-corrected estimates in the \texttt{lspartition} package -- see Footnote \ref{foot:noncentered}.}
This figure showcases that the inequality $E[Y_0|W=0,Y_{-s}=y] \geq E[Y_0|W=1,Y_{-s}=y]$ is satisfied for all $y$, providing evidence in support of our Assumption \ref{a:special} \eqref{a:special:selection}.

Next, recall that the second condition is Assumption \ref{a:special} \eqref{a:special:sd}, which requires that $F_{Y_{-s}|W=0}$ first-order stochastically dominates $F_{Y_{-s}|W=1}$.
We estimate the conditional cumulative distribution functions, $F_{Y_{-s}|W=0}$ and $F_{Y_{-s}|W=1}$, by their empirical counterparts, i.e., empirical CDFs conditionally on $W=0$ and $W=1$, respectively.
Their estimates are plotted in Figure \ref{fig:educ:ass3}.
The solid line indicates the estimates of $F_{Y_{-s}|W=0}$ while the dotted line indicates the estimates of  $F_{Y_{-s}|W=1}$.
Their 95\% confidence intervals are indicated by the shaded regions in colors, although they are nearly invisible due to the large sample size again.
This figure demonstrates that $F_{Y_{-s}|W=0}$ indeed first-order stochastically dominates $F_{Y_{-s}|W=1}$, providing evidence in support of Assumption \ref{a:special} \eqref{a:special:sd}.
%%%%%%%%%%%%%%%%%%%%%%%%%%%%%%%%%%%%%%%
\begin{figure}[t]
\centering
\includegraphics[width=0.5\textwidth]{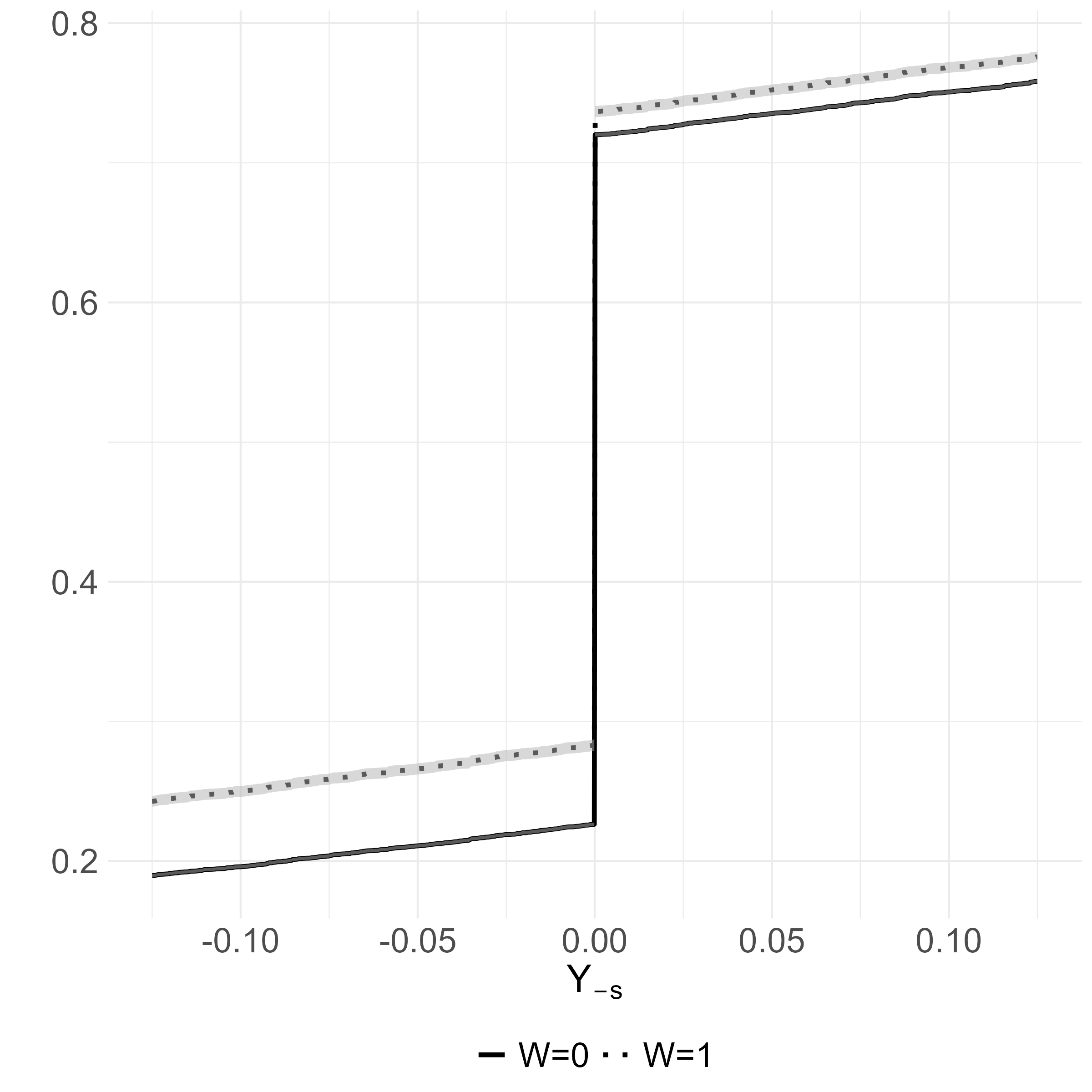}
\caption{Evidence of Assumption \ref{a:special} \eqref{a:special:dec} for the educational program. The conditional expectation function $\Phi$ is non-parametrically estimated by the Nadaraya-Watson estimator. The estimates, along with their 95\% confidence intervals, are plotted.  }${}$
\label{fig:educ:ass3}
\end{figure}
%%%%%%%%%%%%%%%%%%%%%%%%%%%%%%%%%%%%%%%

The third condition is Assumption \ref{a:special} \eqref{a:special:dec}, which requires $y \mapsto \Phi(y) := E[Y_1-Y_0|W=0,Y_{-s}=y]$ to be weakly decreasing.
We estimate the conditional expectation function $\Phi$ non-parametrically by using the partitioning-based least squares regression -- see Footnote \ref{foot:nonparametric_estimation} for details.
The estimates, along with their 95\% confidence bands, are plotted in Figure \ref{fig:educ:ass2}.
The plot shows that this function $\Phi$, is indeed non-increasing, providing evidence in support of of our Assumption \ref{a:special} \eqref{a:special:dec}.

Our analysis in the current section omits auxiliary covariates.
In Appendix \ref{sec:additional:educ}, we provide further evidence in support of our assumption even after accounting for auxiliary covariates.
In particular, accounting for the covariates will strengthen the empirical support of our Assumption \ref{a:special} \eqref{a:special:selection} as mentioned above.

In summary, all the three components, (i)--(iii), of Assumption \ref{a:special} hold robustly for the educational program regardless of whether auxiliary covariates are included or not.
Hence, the double bracketing relationship $\Delta(\attm) \leq \Delta(\attdidm) \leq \Delta(\attdid)$, as predicted by our Theorem \ref{theorem:main}, is expected to hold, as we empirically confirmed in Section \ref{sec:educ}.

%%%%%%%%%%%%%%%%%%%%%%%%%%%%%%%%%%%%%%%
\begin{figure}[t]
\centering
\includegraphics[width=0.55\textwidth]{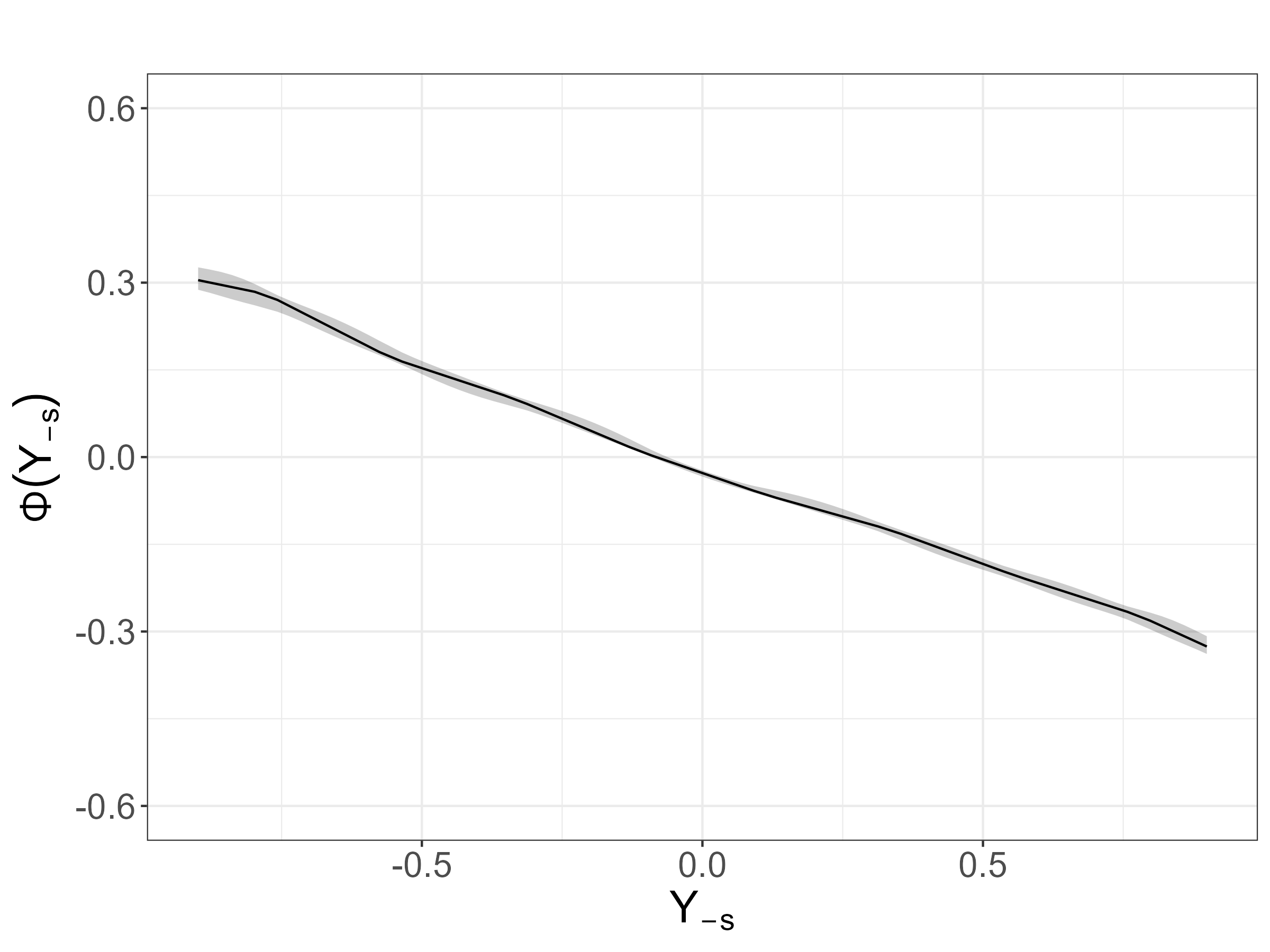}
\caption{Evidence of Assumption \ref{a:special} \eqref{a:special:dec} for the educational program. The conditional expectation function $\Phi$ is non-parametrically estimated by the Nadaraya-Watson estimator. The estimates, along with their 95\% confidence intervals, are plotted.}${}$
\label{fig:educ:ass2}
\end{figure}
%%%%%%%%%%%%%%%%%%%%%%%%%%%%%%%%%%%%%%%

%%%%%%%%%%%%%%%%%%%%%%%%%%%%%%%%%%%%%%%
\section{Summary and Discussions}
%%%%%%%%%%%%%%%%%%%%%%%%%%%%%%%%%%%%%%%
The paper evaluates the relative performance of three estimands -- Matching (M), Difference-in-Differences (DID), and a hybrid method (DIDM) -- for estimating causal effects in observational studies, particularly in the context of job training and educational programs. Our analysis reveals a consistent inequality: \( \attm \leq \att \leq \attdid \). This indicates that Matching tends to produce the most conservative estimates, while DID often yields the most optimistic ones. When selecting a single method, it may be prudent to favor the more conservative Matching estimator. If the Matching estimator suggests a positive effect, this provides a compelling argument that the causal effect is indeed likely positive, leading to a more robust conclusion.

Moreover, these estimands can be utilized in a complementary manner. If practitioners believe that any one of the three identification assumptions -- unconfoundedness (for M), parallel trends (for DID), or conditional parallel trends (for DIDM) -- holds, then the true causal effect is bracketed by \( \attm \leq \att \leq \attdid \). This approach offers a ``triply robust bracketing'' of the causal effect, providing valuable information regardless of which assumption is valid. By applying this framework, researchers can gain a clearer understanding of the potential range of treatment effects, with Matching yielding the most conservative estimate and DID providing the most optimistic one.

Furthermore, the results provide an intuitive strategy for managing uncertainty in identification assumptions. If the Matching (M) estimator indicates a positive effect, the sign of the treatment effect can be interpreted with confidence. The Difference-in-Differences (DID) estimator then provides an upper bound for the magnitude of this effect. This approach simplifies the analysis compared to traditional sensitivity methods (e.g., Manski and Pepper 2018), offering more interpretable bounds for researchers.

\newpage 
\appendix
\section*{Appendix}

%%%%%%%%%%%%%%%%%%%%%%%%%%%%%%%%%%%%%%%
\section{Details of Data}
%%%%%%%%%%%%%%%%%%%%%%%%%%%%%%%%%%%%%%%

This appendix section provides details of the data used in Sections \ref{sec:empirical_double_bracketing} and \ref{sec:empirical_assumption}.

%%%%%%%%%%%%%%%%%%%%%%%%%%%%%%%%%%%%%%%
\subsection{Details of the Data Used in Section \ref{sec:nsw}}\label{sec:appendix:nsw}
%%%%%%%%%%%%%%%%%%%%%%%%%%%%%%%%%%%%%%%
\subsubsection{Data Description}
%%%%%%%%%%%%%%%%%%%%%%%%%%%%%%%%%%%%%%%

Our primary dataset originates from the National Supported Work (NSW) Demonstration, a transitional subsidized work experience program that operated for four years across 15 locations in the United States. This initiative specifically targeted four distinct groups: female long-term AFDC recipients, former drug addicts, ex-offenders, and young school dropouts. Approximately 10,000 individuals took part in the program, each engaging in 12 to 18 months of employment.

The NSW program aimed to assist individuals who faced significant barriers to employment. It provided a structured training environment initially, followed by support in securing regular employment. To ensure the program reached those in genuine need, participants were required to be currently unemployed and to have limited recent employment experience, highlighting the program’s focus on individuals with considerable employment challenges.

A standout feature of the NSW program was its experimental design, which included a randomized control trial at 10 locations between April 1975 and August 1977. In this trial, 6,616 participants were randomly assigned to either a treatment group, which received the program services, or a control group, which did not. Data collection involved a retrospective baseline interview and four follow-up interviews, covering two years before random assignment and up to 36 months afterward. The dataset provides comprehensive information on demographics, employment history, job search behavior, mobility, household income, housing, and drug use.

%%%%%%%%%%%%%%%%%%%%%%%%%%%%%%%%%%%%%%%
\subsubsection{Key Variables}
%%%%%%%%%%%%%%%%%%%%%%%%%%%%%%%%%%%%%%%

In our analysis, we concentrate on the following variables, which are consistent across both the experimental and non-experimental datasets.

The primary outcome of interest is $Y_t$, representing the participants' self-reported earnings. Specifically, we analyze real earnings adjusted to 1982 dollars, in line with the methodology established by \citet{lalonde1986evaluating}.
Next, $W$ serves as a binary indicator of treatment, denoting whether an individual was assigned to the NSW program.
Additionally, we incorporate demographic variables such as race, and education level (indicating high school dropout status) as auxiliary covariates commonly used in M, DIDM, and DID.

The straightforward mean-difference estimate of the NSW program's impact on male participants within the experimental sample is \$886, a figure that is statistically significant at the 10 percent level. This result aligns with the economic disadvantages faced by the NSW population, as evidenced by their low pre-program earnings and the observed decline in earnings from 1974 to 1975, a phenomenon widely recognized in the literature as ``Ashenfelter's dip,'' as discussed in Section \ref{sec:lagged_outcome}.

%\yechany{We might need to be careful of the wordings below, not to make Imbens group(imbens is one of dehejia's advisor) angry}
We chose not to utilize the Dehejia-Wahba (DW) dataset from their 1999 and 2002 studies in our analysis, based on several critical considerations. First, \citet{smith2005does} and \citet{dehejia2005practical} have debated the validity of the DW dataset, particularly questioning its representativeness and the potential biases introduced by the sample restrictions employed. DW exclude approximately 40 percent of the original LaLonde (1986) sample in order to include two years of pre-program earnings data in their model of program participation. This exclusion results in lower mean earnings in 1974 and 1975 for the DW sample compared to the larger LaLonde sample, leading to a significantly different and larger experimental impact estimate of \$1,794, which is more than double that of the LaLonde sample.

Additionally, the data we obtained from the authors of \citet{heckman1998_2matching} includes all pretreatment earnings outcomes, even for the subsample omitted in the DW dataset. This comprehensive dataset removes the need to impose arbitrary sample restrictions, which could otherwise increase sampling uncertainty and potentially bias the results. By employing the full sample, we ensure a broader and more representative analysis, thus upholding the principles of internal and external validity.

%%%%%%%%%%%%%%%%%%%%%%%%%%%%%%%%%%%
\subsection{Details of the Data Used in Section \ref{sec:jtpa}}\label{sec:appendix:jtpa}
%%%%%%%%%%%%%%%%%%%%%%%%%%%%%%%%%%%

%%%%%%%%%%%%%%%%%%%%%%%%%%%%%%%%%%%
\subsubsection{Data Description}
%%%%%%%%%%%%%%%%%%%%%%%%%%%%%%%%%%%
We closely follow the description laid out in \cite{heckman1998characterizing}.
Our primary dataset originates from a randomized evaluation of the Job Training Partnership Act (JTPA) program, conducted across four training centers in the United States. The JTPA program aimed to provide job training and employment services to economically disadvantaged individuals, dislocated workers, and others who faced significant barriers to employment.

The dataset includes information on both experimental treatment and control groups, as well as a non-experimental comparison group of eligible nonparticipants (ENPs) who were located in the same labor markets but chose not to participate in the program at the time of random assignment. Random assignment occurred when individuals applied and were accepted into the JTPA program, ensuring that participants were comparable at the baseline. Control group members were excluded from receiving JTPA services for 18 months after random assignment.

The data collection involved comprehensive surveys administered to all groups, including the ENPs. These surveys captured detailed retrospective information on labor force participation, job spells, earnings, marital status, and other demographic characteristics. In this analysis, we focus on a sample of adult males aged 22 to 54, following \cite{heckman1998characterizing}.

%%%%%%%%%%%%%%%%%%%%%%%%%%%%%%%%%%%
\subsubsection{Key Variables}
%%%%%%%%%%%%%%%%%%%%%%%%%%%%%%%%%%%

In our analysis, we concentrate on the following key variables, which are consistent across both the experimental and non-experimental datasets.

The primary outcome of interest is $Y_t$, representing the participants' earnings. Specifically, we analyze real earnings over a specific period, adjusting for inflation where necessary. The variable $W$ serves as a binary indicator of treatment, denoting whether an individual was assigned to the JTPA program.

Additionally, we include demographic covariates such as sex and age, which are commonly utilized in various econometric models like Matching (M), Difference-in-Differences Matching (DIDM), and Difference-in-Differences (DID). 

\subsection{Details of the Data Used in Section \ref{sec:educ}}\label{sec:appendix:educ}
%%%%%%%%%%%%%%%%%%%%%%%%%%%%%%%%%%%%%%%
\subsubsection{Data Description}\label{sec:data}
%%%%%%%%%%%%%%%%%%%%%%%%%%%%%%%%%%%%%%%

Our primary observational data come from the administrative records of a large urban school district. The dataset includes information on approximately two million children in grades 3 through 8, covering those born between 1966 and 2001.

This dataset encompasses around 15 million test scores in English language arts and math. Due to changes in the testing regime over the past 20 years—such as the transition from district-specific to statewide tests and variations in test timing--we have normalized the test scores by year and grade to have a mean of zero and a standard deviation of one, following established research practices \citep[e.g.,][]{staiger2010searching}. This normalization ensures comparability with other samples across the nation. We also imputed missing test scores using cohort-specific means based on year of birth to account for cohort-level heterogeneity.

%Importantly, following the spirit and methodology of \cite{lalonde1986evaluating} and \cite{heckman1998characterizing}, we use the experimental data from the seminal Project STAR experiment as the baseline truth for $\att$. This same dataset combination has been used in \cite{athey2020combining}. 

%%%%%%%%%%%%%%%%%%%%%%%%%%%%%%%%%%%%%%%
\subsubsection{Key Variables}
%%%%%%%%%%%%%%%%%%%%%%%%%%%%%%%%%%%%%%%

We focus on the following variables in our analysis. The primary outcome of interest is \textbf{$Y_t$}, representing students' test scores, specifically standardized scores that average results from both mathematics and English language arts.

Secondly, \textbf{$W$} is a binary indicator denoting treatment, which in this context refers to the assignment to a small class size.

Lastly, we use gender, race, and eligibility for free lunch to define subpopulations for further analysis.

%%%%%%%%%%%%%%%%%%%%%%%%%%%%%%%%%%%%%%%
\section{Additional Empirical Evidence of the Assumptions}
%%%%%%%%%%%%%%%%%%%%%%%%%%%%%%%%%%%%%%%

In Section \ref{sec:empirical_assumption} in the main text, we examine Assumption \ref{a:special} for each empirical application, focusing on the primary variables without considering auxiliary covariates. In the current appendix section, we provide additional empirical evidence for Assumption \ref{a:special}, now accounting for the auxiliary covariates that were omitted in the main text. 

To incorporate these covariates with minimal shape restrictions, we employ partial linear models for the relevant conditional expectation functions involved in Assumption \ref{a:special}, where the relationship between the main variables is allowed to be non-parametric while the auxiliary variables appear linearly. This approach allows us to use the Nadaraya-Watson estimation after partialing out the auxiliary covariates from both the dependent and main independent variables, similar to the approach used in residual regressions.

%%%%%%%%%%%%%%%%%%%%%%%%%%%%%%%%%%%%%%%
\subsection{Additional Empirical Analyses for Section \ref{sec:nsw}}\label{sec:additional:nsw}
%%%%%%%%%%%%%%%%%%%%%%%%%%%%%%%%%%%%%%%

Figure \ref{fig:nsw:ass1_resid} presents the counterparts of Figure \ref{fig:nsw:ass1}, after partialing out auxiliary covariates.
Observe that the required inequality $E[Y_0|W=0,Y_{-s}=y] \ge E[Y_0|W=1,Y_{-s}=y]$ is still satisfied both for the CPS and PSID data sets, providing robust evidence in support of our Assumption \ref{a:special} \eqref{a:special:selection} even after accounting for the auxiliary covariates.
%%%%%%%%%%%%%%%%%%%%%%%%%%%%%%%%%%%%%%%
\begin{figure}[t]
\centering
CPS\\
\includegraphics[width=0.5\textwidth]{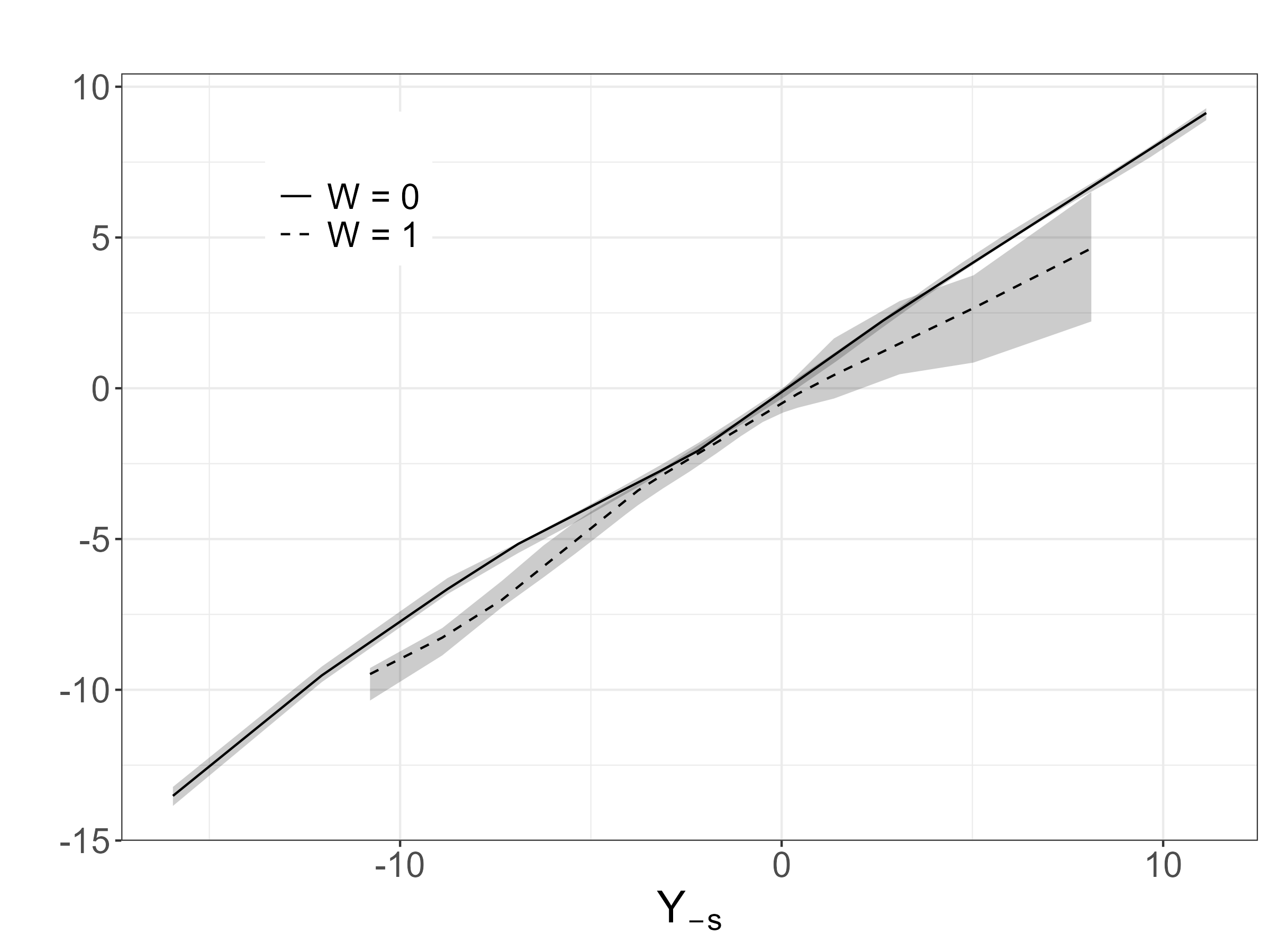}
\\
PSID\\
\includegraphics[width=0.5\textwidth]{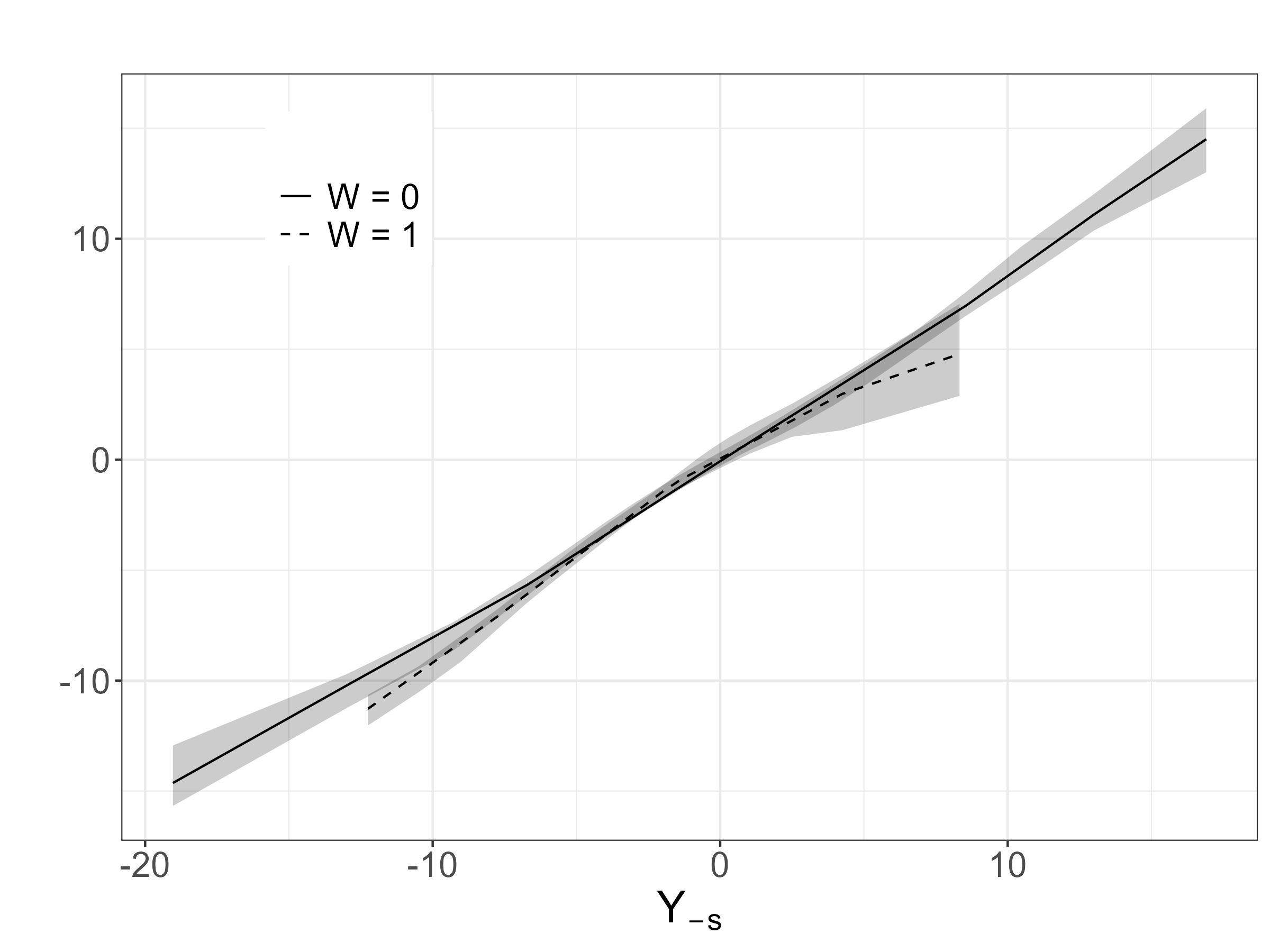}
\caption{Evidence of Assumption \ref{a:special} \eqref{a:special:selection} for the NSW program using the CPS (top) and PSID (bottom) data sets. The solid and dashed lines represent estimates of  the conditional expectation functions $y \mapsto E[Y_0|W=0,Y_{-s}=y]$ and $y \mapsto E[Y_0|W=1,Y_{-s}=y]$, respectively. Shaded areas denote 95\% confidence bands. Both axes are measured in thousands of U.S. dollars. For details on the estimation method, refer to Footnote \ref{foot:nonparametric_estimation}.}${}$
\label{fig:nsw:ass1_resid}
\end{figure}
%%%%%%%%%%%%%%%%%%%%%%%%%%%%%%%%%%%%%%%

Figure \ref{fig:nsw:ass2_resid} presents the counterparts of Figure \ref{fig:nsw:ass2}, after partialing out auxiliary covariates.
Observe that the regression curves are non-increasing for both the CPS and PSID data sets, providing evidence in support of our Assumption \ref{a:special} \eqref{a:special:dec} even after accounting for the auxiliary covariates.
%%%%%%%%%%%%%%%%%%%%%%%%%%%%%%%%%%%%%%%
\begin{figure}[t]
\centering
CPS\\
\includegraphics[width=0.5\textwidth]{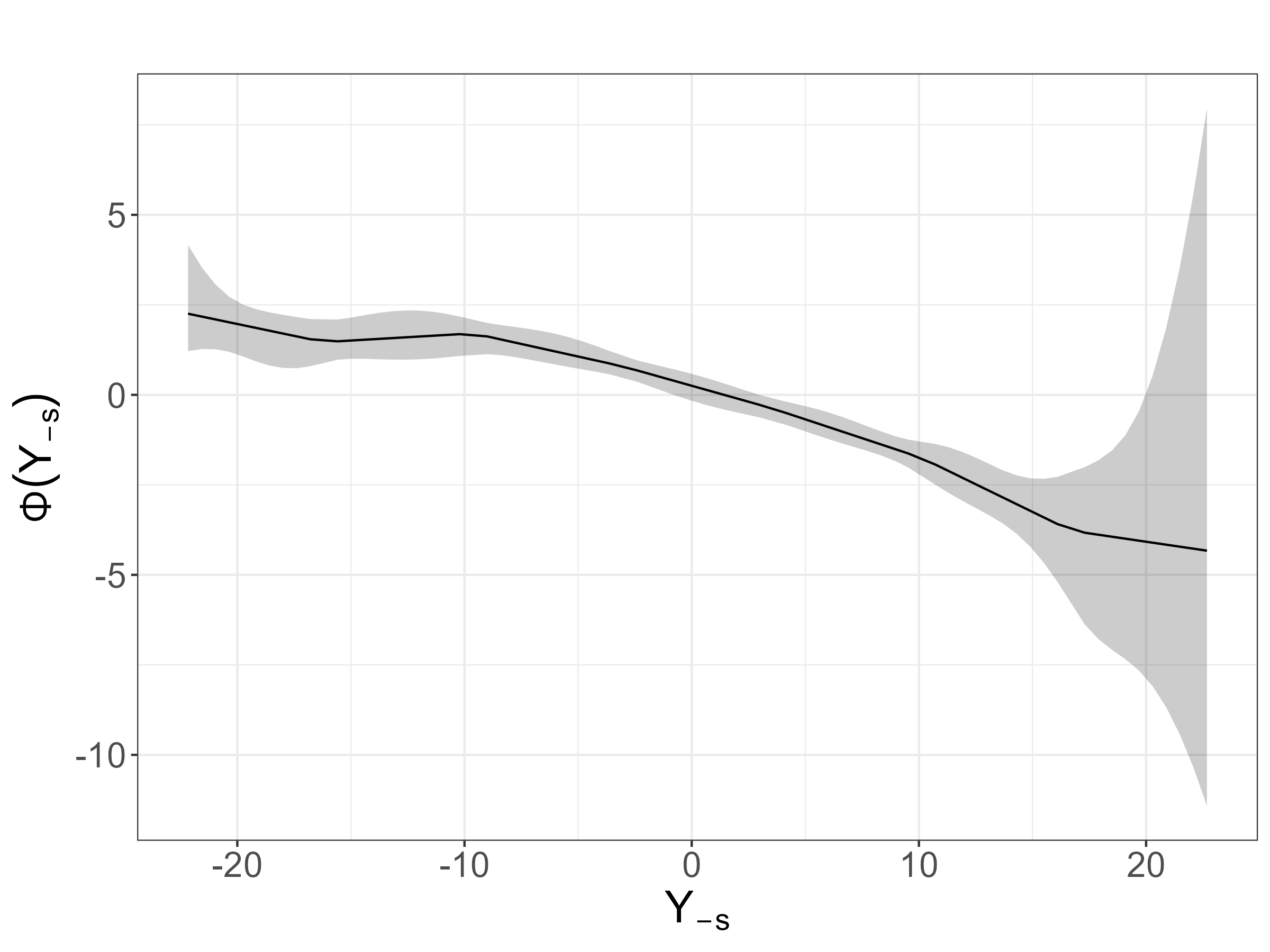}
\\
PSID\\
\includegraphics[width=0.5\textwidth]{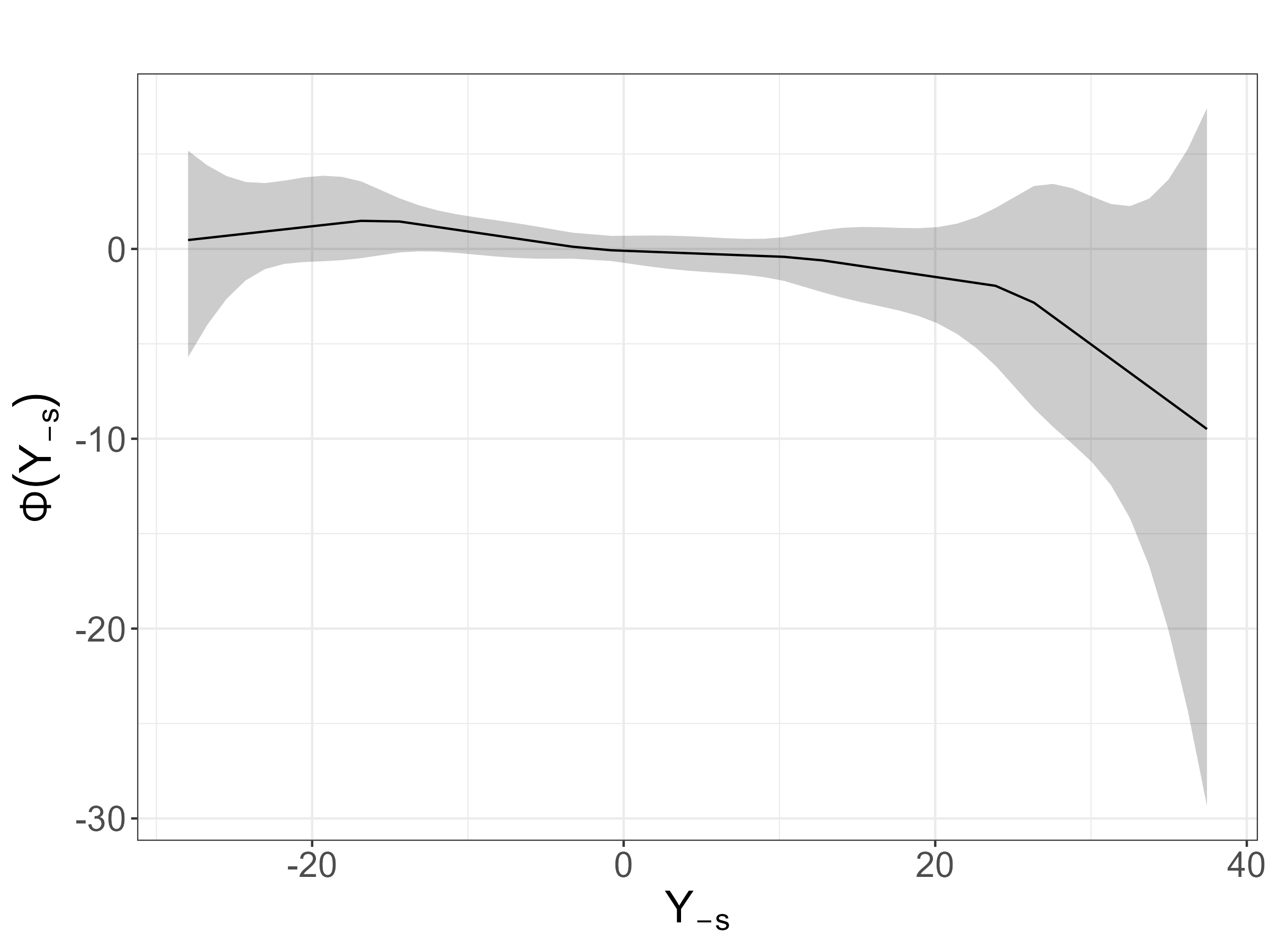}
\caption{Evidence of Assumption \ref{a:special} \eqref{a:special:dec} for the NSW program using the CPS (top) and PSID (bottom) data sets. The conditional expectation function $\Phi$ is estimated non-parametrically by the Nadaraya-Watson estimator after partialing out auxiliary covariates. The estimates, along with their 95\% confidence intervals, are plotted.  Both the vertical and horizontal axes are measured in thousands of U.S. dollars.}${}$
\label{fig:nsw:ass2_resid}
\end{figure}

\newpage${}$\newpage${}$\newpage

%%%%%%%%%%%%%%%%%%%%%%%%%%%%%%%%%%%%%%%
\subsection{Additional Empirical Analyses for Section \ref{sec:educ}}\label{sec:additional:educ}
%%%%%%%%%%%%%%%%%%%%%%%%%%%%%%%%%%%%%%%

Figure \ref{fig:educ:ass1_resid} presents the counterparts of Figure \ref{fig:educ:ass1}, after partialing out auxiliary covariates.
Observe that the regression estimates are non-negative, providing robust evidence in support of our Assumption \ref{a:special} \eqref{a:special:selection} even after accounting for the auxiliary covariates.
%%%%%%%%%%%%%%%%%%%%%%%%%%%%%%%%%%%%%%%
\begin{figure}[t]
\centering
\includegraphics[width=0.5\textwidth]{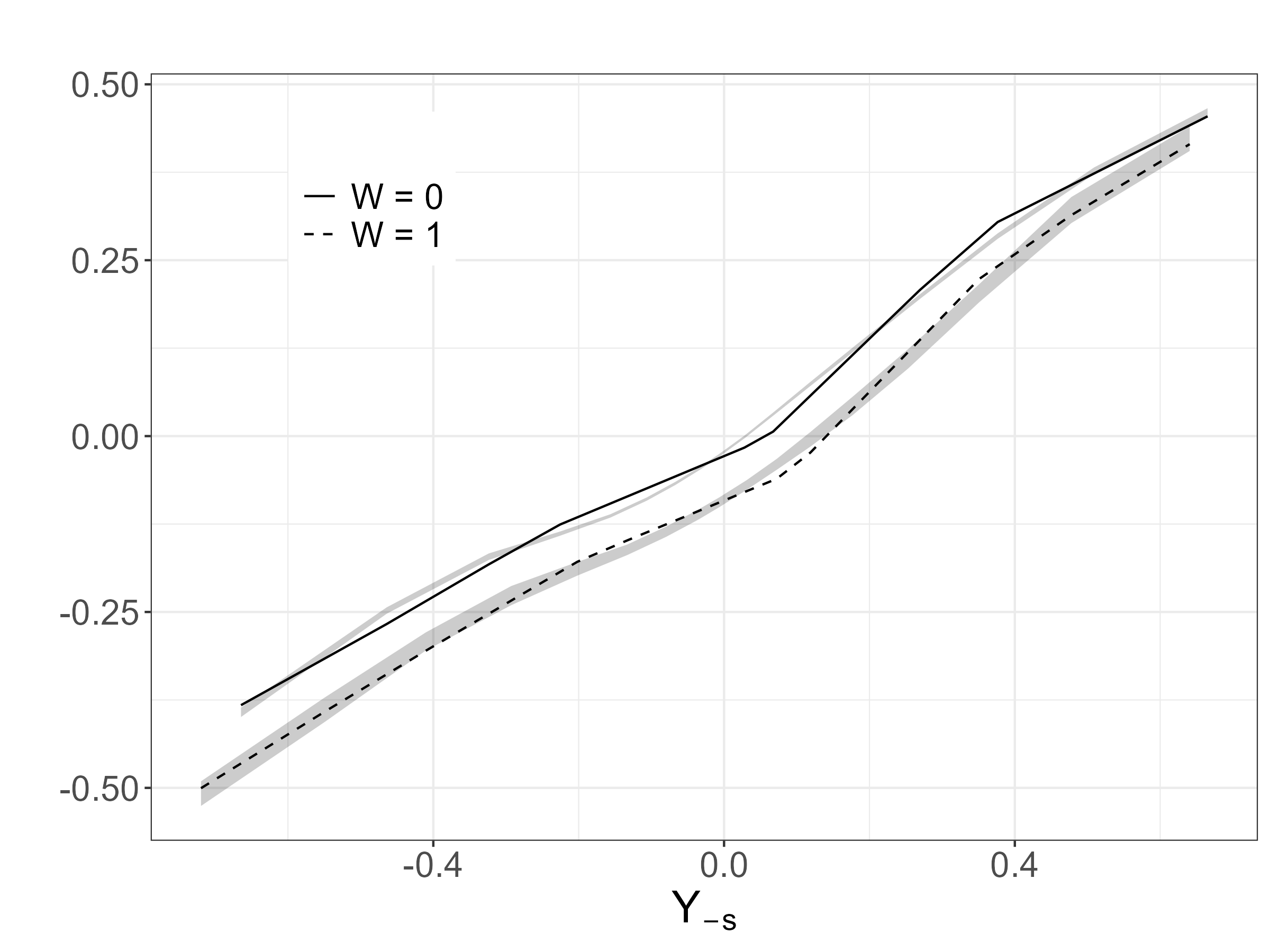}
\caption{Evidence of Assumption \ref{a:special} \eqref{a:special:selection} for the educational program. The solid and dashed lines represent estimates of  the conditional expectation functions $y \mapsto E[Y_0|W=0,Y_{-s}=y]$ and $y \mapsto E[Y_0|W=1,Y_{-s}=y]$, respectively. Shaded areas denote 95\% confidence bands. For details on the estimation method, refer to Footnote \ref{foot:nonparametric_estimation}.}${}$
\label{fig:educ:ass1_resid}
\end{figure}
%%%%%%%%%%%%%%%%%%%%%%%%%%%%%%%%%%%%%%%

Figure \ref{fig:educ:ass2_resid} presents the counterparts of Figure \ref{fig:educ:ass2}, after partialing out auxiliary covariates.
Observe that the regression curves are non-increasing, providing evidence in support of our Assumption \ref{a:special} \eqref{a:special:dec} even after accounting for the auxiliary covariates.
%%%%%%%%%%%%%%%%%%%%%%%%%%%%%%%%%%%%%%%
\begin{figure}[t]
\centering
\includegraphics[width=0.5\textwidth]{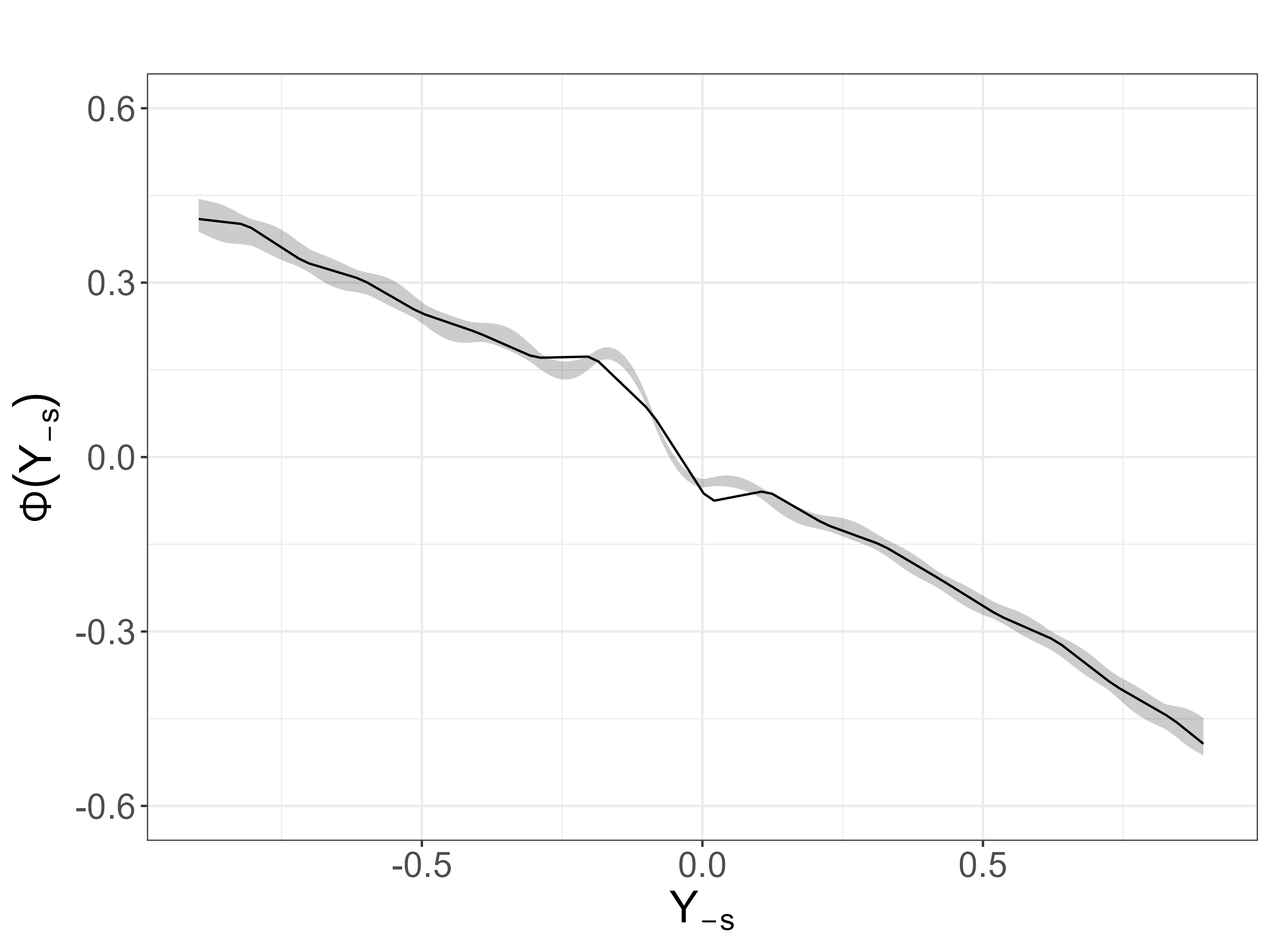}
\caption{Evidence of Assumption \ref{a:special} \eqref{a:special:dec} for the educational program. The conditional expectation function $\Phi$ is non-parametrically estimated by the Nadaraya-Watson estimator. The estimates, along with their 95\% confidence intervals, are plotted.}${}$
\label{fig:educ:ass2_resid}
\end{figure}
%%%%%%%%%%%%%%%%%%%%%%%%%%%%%%%%%%%%%%%

\newpage${}$\newpage

%%%%%%%%%%%%%%%%%%%%%%%%%%%%%%%%%%%%%%%%%%%%%%%%%%%%
\section{Detailed Calculations for Section \ref{sec:parametric}}\label{sec:detailed_calculations}
%%%%%%%%%%%%%%%%%%%%%%%%%%%%%%%%%%%%%%%%%%%%%%%%%%%%

This appendix section presents detailed calculations to derive the expressions \eqref{eq:parametric:m}--\eqref{eq:parametric:did} in Section \ref{sec:parametric}.
We omit the $i$ subscript throughout this appendix section.

Iterated applications of \eqref{eq:parametric} yield
\begin{align*}
E[Y_1|W=1,Y_{-1}] =& \alpha + \beta + \gamma + \delta_1 + \rho (\alpha + \gamma + \delta_0 + \rho Y_{-1}) 
\\
=& (1+\rho)\alpha + \beta + (1+\rho)\gamma + \delta_1 + \rho\delta_0 + \rho^2 Y_{-1}
\qquad{and}
\\
E[Y_1|W=0,Y_{-1}] =& \alpha + \delta_1 + \rho (\alpha + \delta_0 + \rho Y_{-1}) 
\\
=& (1+\rho)\alpha + \delta_1 + \rho\delta_0 + \rho^2 Y_{-1}.
\end{align*}
Substituting these expressions and using the law of iterated expectations yield
\begin{align*}
\attm &= E[ Y_1 | W=1] - E[ E[Y_1 |W=0, Y_{-1} ] | W=1] 
\\
&= E[ E[Y_1|W=1,Y_{-1}] | W=1] - E[ E[Y_1 |W=0, Y_{-1} ] | W=1] 
= \beta + (1+\rho)\gamma.
\end{align*}
This derives the expression in \eqref{eq:parametric:m}.

Next, observe that \eqref{eq:parametric} yields
\begin{align*}
Y_1-Y_0 =& \beta W + (\delta_1 - \delta_0) + \rho (Y_0-Y_{-1}) + \epsilon_1 - \epsilon_0 
\\
=& \rho\alpha + (\beta + \rho\gamma) W + \delta_1 - (1-\rho)\delta_0 - \rho(1-\rho) Y_{-1} + \epsilon_1 - (1-\rho)\epsilon_0,
\end{align*}
where the second equality follows from
$Y_0-Y_{-1} = \alpha + \gamma W + \delta_0 - (1-\rho) Y_{-1} + \epsilon_0$ by \eqref{eq:parametric}.
Thus, we have
\begin{align*}
E[Y_1-Y_0|Y_{-1},W=1] =& \rho\alpha + \beta + \rho\gamma + \delta_1 - (1-\rho)\delta_0 - \rho(1-\rho)Y_{-1}
\qquad\text{and}
\\
E[Y_1-Y_0|Y_{-1},W=0] =& \rho\alpha + \delta_1 - (1-\rho)\delta_0 - \rho(1-\rho)Y_{-1}.
\end{align*}
Substituting these expressions yields
\begin{align*}
\attdidm &= E[E[ Y_1 -Y_0 | Y_{-1}, W=1] - E[ Y_1 -Y_0 | Y_{-1}, W=0]|W=1]
= \beta + \rho\gamma.
\end{align*}
This derives the expression in \eqref{eq:parametric:didm}.

Similarly, we have
\begin{align*}
E[Y_1-Y_0|W=1] =& \rho\alpha + \beta + \rho\gamma + \delta_1 - (1-\rho)\delta_0 - \rho(1-\rho)E[Y_{-1}|W=1]
\qquad\text{and}
\\
E[Y_1-Y_0|W=0] =& \rho\alpha + \delta_1 - (1-\rho)\delta_0 - \rho(1-\rho)E[Y_{-1}|W=0].
\end{align*}
Substituting these expressions yields
\begin{align*}
\attdid &= E[ Y_1 - Y_0 | W= 1] - E[ Y_1 - Y_0 | W=0]
\\
&= \beta + \rho\gamma + \rho(1-\rho) (E[Y_{-1}|W=0] - E[Y_{-1}|W=1]).
\end{align*}
This derives the expression in \eqref{eq:parametric:did}.

\setlength{\baselineskip}{6.8mm}
\bibliographystyle{apalike}
\bibliography{reference}

\def\cprime{$'$}
\begin{thebibliography}{}

\bibitem[Acemoglu et~al., 2019]{acemoglu2019democracy}
Acemoglu, D., Naidu, S., Restrepo, P., and Robinson, J.~A. (2019).
\newblock Democracy does cause growth.
\newblock {\em Journal of Political Economy}, 127(1):47--100.

\bibitem[Angrist and Pischke, 2009]{angrist2009mostly}
Angrist, J.~D. and Pischke, J.-S. (2009).
\newblock {\em Mostly Harmless Econometrics: An Empiricist's Companion}.
\newblock Princeton University Press.

\bibitem[Ashenfelter, 1978]{ashenfelter1978estimating}
Ashenfelter, O. (1978).
\newblock Estimating the effect of training programs on earnings.
\newblock {\em Review of Economics and Statistics}, pages 47--57.

\bibitem[Athey et~al., 2020]{athey2020combining}
Athey, S., Chetty, R., and Imbens, G. (2020).
\newblock Combining experimental and observational data to estimate treatment effects on long term outcomes.
\newblock {\em arXiv preprint arXiv:2006.09676}.

\bibitem[Callaway and Sant'Anna, 2018]{callaway2018difference}
Callaway, B. and Sant'Anna, P.~H. (2018).
\newblock Difference-in-differences with multiple time periods and an application on the minimum wage and employment.
\newblock {\em arXiv preprint arXiv:1803.09015}.

\bibitem[Cattaneo et~al., 2019]{cattaneo2019lspartition}
Cattaneo, M.~D., Farrell, M.~H., and Feng, Y. (2019).
\newblock lspartition: Partitioning-based least squares regression.
\newblock {\em arXiv preprint arXiv:1906.00202}.

\bibitem[Chab{\'e}-Ferret, 2017]{chabe2017should}
Chab{\'e}-Ferret, S. (2017).
\newblock Should we combine difference in differences with conditioning on pre-treatment outcomes?

\bibitem[Chetty et~al., 2014a]{chetty2014measuring1}
Chetty, R., Friedman, J.~N., and Rockoff, J.~E. (2014a).
\newblock Measuring the impacts of teachers i: Evaluating bias in teacher value-added estimates.
\newblock {\em American Economic Review}, 104(9):2593--2632.

\bibitem[Chetty et~al., 2014b]{chetty2014measuring2}
Chetty, R., Friedman, J.~N., and Rockoff, J.~E. (2014b).
\newblock Measuring the impacts of teachers ii: Teacher value-added and student outcomes in adulthood.
\newblock {\em American Economic Review}, 104(9):2633--2679.

\bibitem[De~Chaisemartin and d’Haultfoeuille, 2020]{deChaisemartin2020two}
De~Chaisemartin, C. and d’Haultfoeuille, X. (2020).
\newblock Two-way fixed effects estimators with heterogeneous treatment effects.
\newblock {\em American Economic Review}, 110(9):2964--2996.

\bibitem[De~Chaisemartin and d’Haultfoeuille, 2023]{deChaisemartin2023two}
De~Chaisemartin, C. and d’Haultfoeuille, X. (2023).
\newblock Two-way fixed effects and differences-in-differences with heterogeneous treatment effects: A survey.
\newblock {\em Econometrics Journal}, 26(3):C1--C30.

\bibitem[Dehejia, 2005]{dehejia2005practical}
Dehejia, R. (2005).
\newblock Practical propensity score matching: a reply to smith and todd.
\newblock {\em Journal of Econometrics}, 125(1-2):355--364.

\bibitem[Dehejia and Wahba, 1999]{dehejia1999causal}
Dehejia, R.~H. and Wahba, S. (1999).
\newblock Causal effects in nonexperimental studies: Reevaluating the evaluation of training programs.
\newblock {\em Journal of the American Statistical Association}, 94(448):1053--1062.

\bibitem[Dehejia and Wahba, 2002]{dehejia2002propensity}
Dehejia, R.~H. and Wahba, S. (2002).
\newblock Propensity score-matching methods for nonexperimental causal studies.
\newblock {\em Review of Economics and Statistics}, 84(1):151--161.

\bibitem[Ding and Li, 2019]{ding2019bracketing}
Ding, P. and Li, F. (2019).
\newblock A bracketing relationship between difference-in-differences and lagged-dependent-variable adjustment.
\newblock {\em Political Analysis}, 27(4):605--615.

\bibitem[Dube et~al., 2023]{dube2023local}
Dube, A., Girardi, D., Jorda, O., and Taylor, A.~M. (2023).
\newblock A local projections approach to difference-in-differences event studies.
\newblock Technical report, National Bureau of Economic Research.

\bibitem[Heckman et~al., 1998a]{heckman1998characterizing}
Heckman, J., Ichimura, H., Smith, J., and Todd, P. (1998a).
\newblock Characterizing selection bias using experimental data.
\newblock {\em Econometrica}, 66(5):1017--1098.

\bibitem[Heckman et~al., 1998b]{heckman1998_2matching}
Heckman, J.~J., Ichimura, H., and Todd, P. (1998b).
\newblock Matching as an econometric evaluation estimator.
\newblock {\em Review of Economic Studies}, 65(2):261--294.

\bibitem[Heckman et~al., 1999]{heckman1999economics}
Heckman, J.~J., LaLonde, R.~J., and Smith, J.~A. (1999).
\newblock The economics and econometrics of active labor market programs.
\newblock In {\em Handbook of Labor Economics}, volume~3, pages 1865--2097. Elsevier.

\bibitem[Imai et~al., 2023]{imai2023matching}
Imai, K., Kim, I.~S., and Wang, E.~H. (2023).
\newblock Matching methods for causal inference with time-series cross-sectional data.
\newblock {\em American Journal of Political Science}, 67(3):587--605.

\bibitem[LaLonde, 1986]{lalonde1986evaluating}
LaLonde, R.~J. (1986).
\newblock Evaluating the econometric evaluations of training programs with experimental data.
\newblock {\em American Economic Review}, pages 604--620.

\bibitem[Roth and Sant'Anna, 2023]{roth2023parallel}
Roth, J. and Sant'Anna, P.~H. (2023).
\newblock When is parallel trends sensitive to functional form?
\newblock {\em Econometrica}, 91(2):737--747.

\bibitem[Roth et~al., 2023]{roth2023review}
Roth, J., Sant’Anna, P.~H., Bilinski, A., and Poe, J. (2023).
\newblock What’s trending in difference-in-differences? a synthesis of the recent econometrics literature.
\newblock {\em Journal of Econometrics}, 235(2):2218--2244.

\bibitem[Smith and Todd, 2005]{smith2005does}
Smith, J.~A. and Todd, P.~E. (2005).
\newblock Does matching overcome lalonde's critique of nonexperimental estimators?
\newblock {\em Journal of Econometrics}, 125(1-2):305--353.

\bibitem[Staiger and Rockoff, 2010]{staiger2010searching}
Staiger, D.~O. and Rockoff, J.~E. (2010).
\newblock Searching for effective teachers with imperfect information.
\newblock {\em Journal of Economic Perspectives}, 24(3):97--118.

\end{thebibliography}
%%%%%%%%%%%%%%%%%%%%%%%%%%%%%%%%%%%%%%%

\end{document}